\documentclass[10 pt, journal]{IEEEtran}
\IEEEoverridecommandlockouts
\usepackage{amsmath, amsthm, amssymb, amsfonts, mathtools, xargs, tensor, units, cite, stmaryrd, mathrsfs, comment, graphicx,textcomp}
\usepackage{pgfplots}
\pgfplotsset{compat=newest}

\usetikzlibrary{shapes, arrows.meta, positioning, backgrounds}
\usetikzlibrary{calc,positioning,fit,shapes.multipart,arrows,shadows,backgrounds}

\tikzstyle{block} = [draw, rectangle, minimum height=3em, minimum width=6em, align=center, fill=gray!30, text=black, drop shadow]
\tikzstyle{smallblock} = [draw, rectangle, minimum height=2em, minimum width=2em, align=center, fill=gray!30, text=black, drop shadow]

\newcommand\scalemath[2]{\scalebox{#1}{\mbox{\ensuremath{\displaystyle #2}}}}
\tikzstyle{output} = [coordinate]

\newcommand{\interior}[1]{%
  {\kern0pt#1}^{\mathrm{o}}%
}

\pgfplotsset{/pgfplots/circular legend/.style={
        /pgfplots/legend image code/.code={%
            \fill[##1,/tikz/.cd, ] circle[radius=2pt];
            \draw[##1,/tikz/.cd, solid]
            (0cm,0.cm) -- (0cm,0cm);},},              
}
\usepackage{soul} 
\usepackage{mathrsfs}
\usepackage{multibib}
\theoremstyle{definition}
\newtheorem{defn}{Definition}
\newtheorem{prob}{Problem}
\theoremstyle{definition}
\newtheorem{assum}{Assumption}
\newtheorem{lem}{Lemma}
\newtheorem{thm}{Theorem}
\theoremstyle{remark}
\newtheorem{remark}{Remark}

\usepackage{esint}
\usepackage{commath}
\usepackage{float}
\usepackage{algpseudocode}

\usepackage[inline]{enumitem}

\definecolor{darkgreen}{rgb}{0.0, 0.5, 0.0}

\DeclareMathOperator\cbf{\scalemath{0.7}{\mathrm{cbf}}}

\DeclareMathOperator{\tr}{tr} 
\DeclareMathOperator\sign{sgn}

\DeclareMathOperator\Int{int}
\DeclareMathOperator{\vect}{vec}

\DeclareMathOperator{\proj}{\mathtt{proj}}

\definecolor{color1}{RGB}{0,120,0}      
\definecolor{color2}{rgb}{1.0, 0.22, 0.0} 
\definecolor{color3}{RGB}{0,100,100}    
\definecolor{color4}{RGB}{120,0,120}    
\definecolor{color5}{RGB}{255,100,0}    
\definecolor{color6}{RGB}{0,80,160}     
\definecolor{color7}{RGB}{255,50,50}    
\definecolor{color8}{RGB}{0,160,0}      
\definecolor{color9}{RGB}{200,120,0}    
\definecolor{color10}{RGB}{0,120,200}   
\definecolor{color11}{RGB}{50,0,100}    
\definecolor{color12}{RGB}{200,0,100}   
\definecolor{color13}{rgb}{0.74, 0.72, 0.42} 
\definecolor{color14}{RGB}{0,150,80}    
\definecolor{color15}{RGB}{160,0,0}     

\usepackage[linesnumbered,ruled,vlined]{algorithm2e}

\makeatletter
\def\ps@IEEEtitlepagestyle{%
  \def\@oddhead{\hbox{}\scriptsize
  \textit{This work has been submitted to the IEEE for possible publication. 
  Copyright may be transferred without notice, after which this version may no longer be accessible.}
  \hfill}%
  \def\@oddfoot{}%
}
\makeatother
\begin{document}

\title{Safe Output-Feedback Adaptive Optimal Control of Affine Nonlinear Systems}
\author{Tochukwu E. Ogri, Muzaffar Qureshi, Zachary I. Bell, Wanjiku A. Makumi, and Rushikesh Kamalapurkar
\thanks{This research was supported in part by the Air Force Research Laboratory under contract numbers FA8651-24-1-0019 and FA8651-23-1-0006 and the Office of Naval Research under contract number N00014-21-1-2481. Any opinions, findings, or recommendations in this article are those of the author(s), and do not necessarily reflect the views of the sponsoring agencies.}
\thanks{T. E. Ogri, M. Qureshi, and R. Kamalapurkar are with the Department of Mechanical and Aerospace Engineering, University of Florida, Gainesville, FL 32611, USA (e-mail: tochukwu.ogri@ufl.edu; muzaffar.qureshi@ufl.edu; rkamalapurkar@ufl.edu).}
\thanks{Z. I. Bell and Wanjiku A. Makumi are with the Air Force Research Laboratory, Eglin AFB, FL 32542, USA (e-mail: zachary.bell.10@us.af.mil; wanjiku.makumi@us.af.mil).}
}

\maketitle

\begin{abstract} 
In this paper, we develop a safe control synthesis method that integrates state estimation and parameter estimation within an adaptive optimal control (AOC) and control barrier function (CBF)-based control architecture. The developed approach decouples safety objectives from the learning objectives using a CBF-based guarding controller where the CBFs are robustified to account for the lack of full-state measurements. The coupling of this guarding controller with the AOC-based stabilizing control guarantees safety and regulation despite the lack of full state measurement. The paper leverages recent advancements in deep neural network-based adaptive observers to ensure safety in the presence of state estimation errors. Safety and convergence guarantees are provided using a Lyapunov-based analysis, and the effectiveness of the developed controller is demonstrated through simulation under mild excitation conditions.
\end{abstract}

\section{Introduction}\label{section:introduction}
Safety is crucial in the design of autonomous systems operating in uncertain and complex environments. The lack of full-state information makes safe control design challenging since the controller must rely on state estimates to carry out necessary tasks safely despite state estimation errors. The objective of this paper is to develop a safe output-feedback adaptive optimal control architecture to enable simultaneous learning and execution of optimal controllers under safety constraints.

Safe adaptive optimal control (SAOC) of autonomous systems can be realized via adaptive dynamic programming (ADP) (cf. \cite{SCC.Yang.Vamvoudakis.ea2019, SCC.Greene.Deptula.ea2020, SCC.Mahmud.Nivison.ea2021, SCC.Cohen.Belta.ea2020, SCC.Deptula.Chen.ea2020, SCC.Marvi.Kiumarsi.ea2020}), where the value function is learned and used for the synthesis of the controller in real-time. To guarantee safety, results such as
\cite{SCC.Cohen.Belta.ea2020, SCC.Deptula.Chen.ea2020, SCC.Marvi.Kiumarsi.ea2020} incorporate a \emph{barrier-like term} in the cost function of the optimal control problem. However, this approach can lead to a trade-off between safety and performance, with overly restrictive safety conditions hindering the ability of the system to achieve its stability objectives. To avoid this trade-off, \cite{SCC.Cohen.Belta.ea2023} and \cite{SCC.Cohen.Serlin.ea2023} propose a decoupled safety-aware controller by developing a dedicated safety controller which is paired with an ADP controller \cite{SCC.Kamalapurkar.Rosenfeld.ea2016, SCC.Kamalapurkar.Walters.ea2016, SCC.Kamalapurkar.Walters.ea2016a}, allowing the system to maintain safety without compromising stability. However, the results in \cite{SCC.Cohen.Belta.ea2023} rely on an $L_{\infty}$ bound on the gradient of the Lyapunov-like barrier function to guarantee stability, which may be difficult to obtain near the boundary of the safe set. While \cite{SCC.Cohen.Serlin.ea2023} relaxes this bound by using zeroing barrier functions, the controllers in \cite{SCC.Cohen.Belta.ea2023} and \cite{SCC.Cohen.Serlin.ea2023} both require full state feedback. 

Another popular approach to the SAOC problem involves nonlinear coordinate transformations of the system states, initially introduced in \cite{SCC.Graichen.Petit.ea2009} and further explored by \cite{SCC.Yang.Vamvoudakis.ea2019, SCC.Greene.Deptula.ea2020, SCC.Mahmud.Nivison.ea2021} for nonlinear control-affine systems. This technique utilizes barrier functions to transform a state-constrained optimal control problem into an equivalent unconstrained one, but only for problems with box constraints. Recent SAOC techniques like \cite{SCC.Almubarak.Theodorou.ea2021} and \cite{SCC.Bandyopadhyay.Bhasin.ea2025} solve the constrained optimal control problem by utilizing Karush-Kuhn-Tucker (KKT) conditions. This approach relies on the estimation of an optimal state-dependent Lagrange multiplier. The approaches discussed so far require full state information and do not admit straightforward extensions to the partially observed case. Since full state observations are not available in most real systems, the SAOC methods that admit partial observations are critical for real-world applicability.

The partially observed adaptive optimal control problem is considered without safety constraints in results such as \cite{SCC.Yang.Liu.ea2014,SCC.Kamalapurkar2017a,SCC.Self.Harlan.ea2019}, and with safety constraints in results such as \cite{SCC.Mahmud.Abudia.ea2024}. A limitation of the output-feedback control approach in \cite{SCC.Yang.Liu.ea2014} is the computation of the system state using the pseudoinverse of the measurement matrix, which implicitly restricts the technique to systems that have at least as many outputs as states. The applicability of \cite{SCC.Kamalapurkar2017a} and \cite{SCC.Self.Harlan.ea2019} is limited by the need for the system to be in the Brunovsky canonical form. The method in \cite{SCC.Mahmud.Abudia.ea2024} is limited to box constraints and outputs that are selections of state variables.

In this paper, a method for online real-time safe learning that is \emph{robust to modeling errors and relies on partial state information at run-time} is developed for control-affine nonlinear systems with a linear measurement model via integration of state estimation, parameter estimation, and robustified control barrier functions (CBFs) \cite{SCC.Wieland.Allgower.ea2007, SCC.Ames.Grizzle.ea2014,SCC.Agrawal.Panagou.ea2022} in a SAOC framework. An online solution is derived to ensure safety throughout the learning phase. Motivated by the results in \cite{SCC.Jankovic.ea2018, SCC.Cohen.Belta.ea2023, SCC.Cohen.Serlin.ea2023}, we decouple the safety objective from the learning objective to avoid the need for \textit{a priori} unverifiable conditions to guarantee safety. Motivated by the results in \cite{SCC.Agrawal.Panagou.ea2022}, the CBFs are robustified to ensure safety in the presence of state estimation errors. To estimate the states of the system, we leverage recent developments in deep
neural network (DNN)-based adaptive state
estimation \cite{SCC.Joshi.Chowdhary.ea2019, SCC.Joshi.Virdi.ea2020, SCC.Zegers.Sun.ea2023, SCC.Le.Patil.ea2024}, where
the weights of the outermost layers in the DNN are adapted in real-time using an integral concurrent learning (ICL)-based update law while the inner layer DNN weights are adapted intermittently using batch updates.

A na\"{i}ve incorporation of state estimators in ADP is challenging because the resulting estimation errors can propagate through value function updates, leading to instability or suboptimal control policies. The safety constraints further complicate the incorporation of state estimators in ADP since constraint satisfaction must be guaranteed despite estimation errors, which necessitates robust formulations. In this paper, we address the aforementioned challenges by developing an SAOC framework that explicitly accounts for state estimation errors through a high-fidelity DNN-based observer, while ensuring safety using a robust CBF and stability-aware learning updates. 

Robustification of CBFs to ensure safety in the presence of estimation errors typically requires input-to-state stable or bounded-error observers \cite{SCC.Agrawal.Panagou.ea2022} with known error bounds. To develop such observers and to compute the resulting error bounds, one typically assumes that the system state remains in a known compact set \cite{SCC.Agrawal.Panagou.ea2022, SCC.Cosner.Singletary.ea2021}, creating an apparent interdependency in the design of the robust CBF and the observer. In this paper, we resolve this apparent interdependency by developing a control law that, irrespective of the error bounds used in the development of the robust CBF, keeps the system states confined to a compact set that is a superset of the desired safe set. We then compute the observer error bounds and subsequently, the robustified CBFs, using the said compact set to also ensure that the system states remain confined to the safe set. To demonstrate the effectiveness of the proposed approach, we conduct two simulation studies showing that without the developed technique, the controllers fail to guarantee safety. 


\section{Problem Formulation and Preliminaries}	
\label{section:problemFormulation}
\subsection{System Description}
This paper considers continuous time, nonlinear systems of the form 
\begin{subequations}
\label{eq:dynamics}
\begin{align}
\dot{x} &= f(x)+g(x)u,\\ 
y &= Cx, 
\end{align}
\end{subequations}
where $x \in \mathbb{R}^{n}$ is the system state vector, $u \in \mathbb{R}^{m}$ is the control input vector,  $C \in \mathbb{R}^{q \times n}$ is the output matrix, and $y \in \mathbb{R}^{q}$ is the measured output vector. The unknown function $f: \mathbb{R}^{n} \to \mathbb{R}^{n}$ and the known function $g: \mathbb{R}^{n} \to \mathbb{R}^{n \times m}$, denote the drift dynamics and the control effectiveness, respectively. 
\begin{assum}\label{ass:boundedfunctions}
The functions $x \mapsto f(x)$ and $x \mapsto g(x)$ are locally Lipschitz continuous, and satisfy $f(0) = 0$ and $0 < \|g(x)\| \leq \overline{g}$ for some $\overline{g} > 0$ and for all $x \in \mathbb{R}^{n}$.
\end{assum}

To develop the notion of safety, this paper considers a \emph{safe set} $\mathcal{S} \subset \mathbb{R}^{n}$, defined as the zero super level set of a continuously differentiable function $h: \mathbb{R}^{n} \to \mathbb{R}$ such that
\begin{align}
    \mathcal{S} = \{x \in \mathbb{R}^{n} \mid h(x)\geq 0\}\label{eq:safeSet1}, \\
    \partial\mathcal{S} = \{x \in \mathbb{R}^{n} \mid h(x) = 0\}, \\
    \Int(\mathcal{S}) = \{x \in \mathbb{R}^{n} \mid h(x) > 0\}, \label{eq:safeSet3}
\end{align}
 where $\partial\mathcal{S}$ and $\Int(\mathcal{S})$ represent the boundary and interior of $\mathcal{S}$, respectively. Given a locally Lipschitz continuous feedback control policy $\pi: \mathbb{R}^{n} \times \mathbb{R}_{\geq 0}
 \to \mathbb{R}^{m}
$, the closed-loop dynamics can be expressed as
\begin{equation} \label{eq:closedSystem}
\dot{x} = f_{\operatorname{cl}}(x, t) \coloneqq f(x) + g(x)\pi(x, t), \quad x\left(0\right) = x_{0},
\end{equation}
where $f_{\operatorname{cl}} : \mathbb{R}^{n} \times \mathbb{R}_{\geq 0} \to \mathbb{R}^{n}$ is also locally Lipschitz continuous. Let $t \mapsto x(t)$ be a solution to the system in \eqref{eq:closedSystem}, starting from initial condition $x_{0}$ and under the control policy $(x, t) \mapsto \pi(x, t)$. 
\begin{defn}\label{defn:safety}
    The system in \eqref{eq:closedSystem} is safe with respect to the sets $\left(\mathcal{S}_{0}, \mathcal{S}\right)$ if $x_{0} \in \mathcal{S}_{0}$ implies $x(t) \in \mathcal{S}$ for all $t \in \mathbb{R}_{\geq 0}$.
\end{defn}
 The objective is to design an observer to estimate the states of the system online, using input-output measurements, and to simultaneously synthesize and utilize a feedback control policy $(x, t) \mapsto \pi(x, t)$ that minimizes a cost functional (introduced in \eqref{eq:costFunctional}) while maintaining safety with respect to the sets $\left(\mathcal{S}_{0}, \mathcal{S}\right)$.
\begin{figure}
\centering
\begin{tikzpicture}[scale=0.5, transform shape]

    \node [minimum height=1.5cm, minimum width=1.75cm, draw, align=center, fill=blue!5, text=black, rounded corners] (controller) {$u = \pi(\hat{x}, t)$};
    \node [minimum height=1cm, minimum width=3cm, above=0.5cm of controller, anchor=north] (controllerTitle) {};
    \node [draw=blue!80!black, thick, dashed, fit={(controllerTitle) (controller)}, rounded corners] (controllerBox) {};
    \node [below, blue!80!black, align=center] at (controllerBox.north) {\textbf{CBF-RL Policy}};

    \node [minimum height=2cm, minimum width=1cm, draw, align=center, fill=red!5, text=black, rounded corners, right=2cm of controller] (system) {$\dot{x} = f(x) + g(x)u$};
    \node [above=0.5cm of system, anchor=north] (systemTitle) {};
    \node [draw=red!80!black, thick, dashed, fit={(systemTitle) (system)}, rounded corners] (systemBox) {};
    \node [below, red!80!black, align=center] at (systemBox.north) {\textbf{System}};

    \node [coordinate, right=0.5cm of system] (output) {};
    \node [minimum height=2em, minimum width=2em, draw, align=center, fill=gray!10, text=black, rounded corners, right=0.5cm of output] (outputblock) {$y = Cx$};
    \node [above=0.5cm of outputblock, anchor=north] (outputTitle) {};
    \node [draw=gray!70!black, thick, dashed, fit={(outputTitle) (outputblock)}, rounded corners] (outputBox) {};
    \node [below, gray!70!black, align=center] at (outputBox.north) {\textbf{Output}};
    
    \node [coordinate, right=1.5cm of outputblock] (outputy) {};

    \node [minimum height=2cm, minimum width=1cm, draw, align=center, fill=green!10, text=black, rounded corners, below=1.25cm of system, xshift=1cm] (stateEstimator) {$\dot{\hat{x}} = \hat{f}(\hat{x}, \hat{\theta}) + g(\hat{x})u + K(y - C\hat{x})$};
    \node [above=0.5cm of stateEstimator, anchor=north] (stateEstimatorTitle) {};
    \node [draw=green!50!black, thick, dashed, fit={(stateEstimatorTitle) (stateEstimator)}, rounded corners] (stateEstimatorBox) {};
    \node [below, text=green!50!black, align=center] at (stateEstimatorBox.north) {\textbf{State Observer}};

    \node [minimum height=1cm, minimum width=3cm, draw, align=center, fill=orange!10, text=black, rounded corners, below=1.5cm of stateEstimator] (paramEstimator) {$\dot{\hat{\theta}} = f_{\theta}(\hat{\theta}, t)$};
    \node [above=0.5cm of paramEstimator, anchor=north] (paramEstimatorTitle) {};
    \node [draw=orange!70!black, thick, dashed, fit={(paramEstimatorTitle) (paramEstimator)}, rounded corners] (paramEstimatorBox) {};
    \node [below, orange!70!black, align=center] at (paramEstimatorBox.north) {\textbf{Parameter Estimator}};

    \node [
    minimum height=2cm, 
    minimum width=2cm, 
    draw, 
    align=center, 
    fill=purple!10, 
    text=black, 
    rounded corners,
    below=2.75cm of controller, 
    xshift=-0.3cm, 
    rectangle split, 
    rectangle split horizontal, 
    rectangle split parts=2
    ] (actor) {
        $\hat{u}(\hat{x}, \hat{W}_{a})$ 
        \nodepart{two} 
        $\hat{V}(\hat{x}, \hat{W}_{c})$
    };
    \node [above=0.5cm of actor, anchor=north] (actorTitle) {};
    \node [draw=purple!70!black, thick, dashed, fit={(actorTitle) (actor)}, rounded corners] (actorBox) {};
    \node [below, purple!70!black, align=center] at (actorBox.north) {Actor + Critic};

    \draw [->, thick] (controller) -- node[name=u, pos=0.5, above] {$u$} (system);
    \draw [->, thick] (system) -- node[name=systemoutput, above] {$x$} (outputblock);
    \draw [->, thick] (outputblock) -- node[name=y, above] {$y$} (outputy);
    \draw [->, thick] ([xshift=-1.0cm]outputy) |-  ([yshift=0.5cm]stateEstimator);
    \draw [->, thick] ([xshift=-1.875cm]outputy) ++(0.75,0) |- ++(0,2.5) -| ([yshift=0.25cm, xshift=-1cm]controller.west) |- ([yshift=0.25cm]controller.west);
    \draw [->, thick] ([yshift=0.8cm]stateEstimator.west) -- ++(-0.5,0) -| node[pos=0.25, above] {$\hat{x}$} ([xshift=0.6cm]controller.south);
    \draw [->, thick] (paramEstimator.west) -- ++(0, 0) -| node[pos=0.5, left] {$\hat{\theta}$} ([xshift=-1cm, yshift=-0.25cm]stateEstimator.west) |- ([yshift=-0.25cm]stateEstimator.west);
    \draw [->, thick] ([yshift=-0.2cm]stateEstimator.east) -- ++(1.5,0) |- node[pos=0.75, above] {$\hat{x}$} (paramEstimator.east);
    \draw [->, thick] (actor.west) -- ++(-0.5, 0) |- node[pos=0.25, left] {$\pi_{\operatorname{des}}(\hat{x}, t)$} ([yshift=-0.25cm]controller.west);
    \draw [->, thick] (stateEstimator.south) -- ++(0,0) |- node[pos=0.9, above] {$\hat{x}$} ([yshift=-0.25cm]actor.east);
    \draw [->, thick] ([xshift=1cm] controller.east) -- ++(0,-1.35) -| ([xshift=-1.25cm]stateEstimator.north);

\end{tikzpicture}
\caption{Schematic of the proposed control system showing the integration of parameter estimation, state estimation, RL, and CBFs to enable safe adaptive optimal control.}
\end{figure}
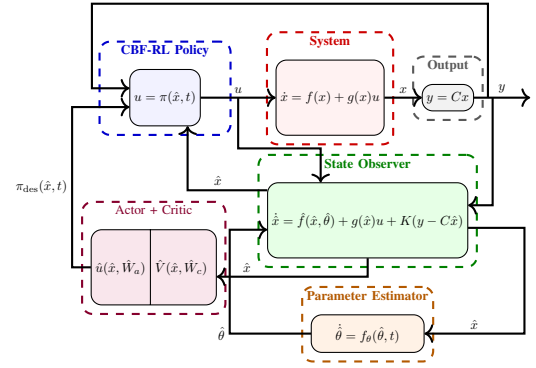
 \subsection{Control Barrier Functions}

The safety objective motivates the adoption of control barrier functions to formulate conditions to identify the set of control inputs that guarantee the system in \eqref{eq:closedSystem} is safe with respect to the sets $\left(\mathcal{S}_{0}, \mathcal{S}\right)$.
\begin{defn}\label{defn:cbf}
 \cite{SCC.Ames.Xu.ea2017} Given a set $\mathcal{S}\subset \mathbb{R}^{n}$, as defined in \eqref{eq:safeSet1}--\eqref{eq:safeSet3}, a continuously differentiable function $h : \mathbb{R}^{n} \to \mathbb{R}$ is a control barrier function (CBF) for  \eqref{eq:dynamics}, if there exist a class $\mathcal{K}$ function $\alpha: \mathbb{R}_{\geq 0} \to \mathbb{R}_{\geq 0}$ such that
\begin{equation}
    \sup_{u \in \mathbb{R}^{m}}\left\{\nabla_{x}h(x)\left(f(x) + g(x)u\right)\right\} \geq -\alpha\left(h(x)\right),
\end{equation}
for all $x \in \mathbb{R}^{n}$, where $\nabla_{(\cdot)} \coloneqq  \frac{\partial}{\partial (\cdot)}$ denotes the partial derivative operator.
\end{defn}
  The following result on safe control design motivates the constraint utilized in Problem~\ref{prob:coop} to guarantee safety.
\begin{thm}\label{thm:forwardInvariantC}
    \cite[Corollary~2]{SCC.Ames.Xu.ea2017} Given a set $\mathcal{S} \subset \mathbb{R}^{n}$ and a CBF $h : \mathbb{R}^{n} \to \mathbb{R}$ for the system in \eqref{eq:dynamics}, if $\nabla_{x}h(x)g(x) \neq 0$ for all $x \in \partial \mathcal{S}$, then any locally Lipschitz continuous control policy $\pi:\mathbb{R}^{n} \times \mathbb{R}_{\geq 0} \to \mathbb{R}^{m}$ such that $\pi(x, t) \in \mathcal{K}_{\cbf}(x)$ for all $x \in \mathbb{R}^{n} \text{ and } t \in \mathbb{R}_{\geq 0}$, where
\begin{equation}\label{eq:cbfControlSet}
    \mathcal{K}_{\cbf}(x) \coloneqq \{u \in \mathbb{R}^{m}\!: \!\nabla_{x}h(x)\left(f(x) \!+\! g(x)u\right) \geq - \alpha\left(h(x)\right) \},
\end{equation}
guarantees the system in \eqref{eq:closedSystem} is safe with respect to the sets $(\mathcal{S}, \mathcal{S})$.
\end{thm}
\begin{prob}{Constrained optimal control problem}\label{prob:coop}
\begin{equation}\label{eq:opimalProb}
\begin{aligned}
\min_{u: \mathbb{R}_{\geq 0} \to \mathbb{R}^{m}} \quad & J\left(x_{0}, u(\cdot)\right),\\
\textrm{s.t.} \quad & u(t) \in \mathcal{K}_{\cbf}(x(t)), \forall t \in \mathbb{R}_{\geq 0},
\end{aligned}
\end{equation}
 where the cost functional $J$ is defined as
 \begin{equation} \label{eq:costFunctional}
J\left(x_{0}, u(\cdot)\right)  \coloneqq \int_{0}^\infty r\left(x \left(\tau\right), u \left(\tau\right)\right)  \mathrm{d}\tau,
\end{equation}
$x(\cdot)$ denotes the trajectory of \eqref{eq:dynamics} starting from $x_{0}$ and under the control $u(\cdot)$, $r: \mathbb{R}^{n} \times \mathbb{R}^{m} \to \mathbb{R}$ is the instantaneous cost defined as
\begin{equation}\label{eq:cost} r\left(x ,u \right) \coloneqq Q(x) + u^\top R u,
\end{equation}
 $Q: \mathbb{R}^{n} \to \mathbb{R}$ is a continuous positive definite
function, and $R \in \mathbb{R}^{m \times m}$ is a symmetric positive definite control penalty matrix. The function $Q$ is selected such that there exists a class $\mathcal{K}_\infty$ function $\underline{q}: \mathbb{R}_{\geq 0} \to \mathbb{R}_{\geq 0}$ such that
\begin{equation}\label{eq:qBound}
    \underline{q}\left(\|x\|\right) \leq Q(x), \forall x \in\mathbb{R}^n.
\end{equation} 
\end{prob}
 
Even with perfect state information, solving the constrained optimal control problem outlined in Problem~\ref{prob:coop} for the nonlinear system in \eqref{eq:dynamics} is challenging (cf. \cite{SCC.Almubarak.Theodorou.ea2021, SCC.Bandyopadhyay.Bhasin.ea2025}). While \cite{SCC.Almubarak.Theodorou.ea2021} successfully tackles the CBF-constrained optimal control problem by minimizing the constrained generalized Hamilton-Jacobi Bellman (GHJB) equation, the modified Galerkin successive approximation approach in \cite{SCC.Almubarak.Theodorou.ea2021} becomes difficult to implement if only partial state measurements are available. In this paper, we decouple the safety objective from the stability objective, which, while compromising optimality, guarantees safety regardless of learning performance \cite{SCC.Cohen.Belta.ea2023}.

\section{Robust safety}\label{section:safeQPControl}

 To meet the hypothesis of Theorem~\ref{thm:forwardInvariantC}, a control policy $\pi$ is needed that, for all $x\in \mathbb{R}^{n}$ and $t \in \mathbb{R}_{\geq 0}$, satisfies $\pi(x, t) \in \mathcal{K}_{\cbf}(x)$. Since $x$ is not available for feedback, one can only ensure that $\pi(\hat{x}, t) \in \mathcal{K}_{\cbf}(\hat{x})$ for all $\hat{x}$ and $t$.
 The following theorem, adapted from \cite{SCC.Agrawal.Panagou.ea2022}, provides safety guarantees in the presence of state estimation errors. 
\begin{thm}\label{thm:robustSafety}
  \cite{SCC.Agrawal.Panagou.ea2022} Let $\varepsilon \in (0, \varepsilon^*)$ be a user-defined constant, where $\varepsilon^{*} \coloneqq \sup \left\{ \varepsilon > 0 : \exists \hat{x} \in \mathcal{S},\ \bar{B}(\hat{x}, \varepsilon) \subseteq \mathcal{S} \right\}$\footnote{$\bar{B}(x, y)$ denotes a closed ball of radius $y$ centered at $x$.}.
    If the set $\mathcal{S}$ is defined by the CBF $h: \mathbb{R}^{n} \to \mathbb{R}$, Assumptions~\ref{ass:boundedfunctions}--\ref{ass:boundedGradientG} are satisfied, the state estimates are initialized such that $\hat{x}_{0} \in \{\hat{x}: \bar{B}(\hat{x}, \varepsilon) \subset \mathcal{S}\}$, the control policy $\pi: \mathcal{S} \times \mathbb{R}_{\geq 0} \to \mathbb{R}^{m}$ is locally Lipschitz continuous in the first argument, piecewise-continuous in the second argument, and satisfies  
  \begin{equation}\label{eq:controlCondition}
      \pi(\hat{x}, t) \in \bigcap_{x \in \bar{B}(\hat{x}, \varepsilon)}{\mathcal{K}_{\cbf}}(x), \, \forall x, \hat{x} \in \mathcal{S}, \text{and } t \in \mathbb{R}_{\geq 0},
  \end{equation}
and if the observer used to generate $\hat{x}$ ensures that for all $x_{0} \in (\bar{B}(\hat{x}_0,\varepsilon), \mathcal{S})$, the resulting trajectory $x$ satisfies $x(t) \in \bar{B}(\hat{x}(t), \varepsilon)$ for all $t$, then the controller $u(t) = \pi(\hat{x}(t), t)$ renders the closed loop system in \eqref{eq:closedSystem} safe with respect to the sets $\left(\bar{B}(\hat{x}_{0}, \varepsilon),  \mathcal{S}\right)$.
\end{thm}
\begin{remark}
    A \emph{bounded-error} state observer that ensures that for all $x_{0} \in (\bar{B}(\hat{x}_0,\varepsilon), \mathcal{S})$, the resulting trajectory $x$ satisfies $x(t) \in \bar{B}(\hat{x}(t), \varepsilon)$ for all $t \in \mathbb{R}_{\geq 0}$ under the controller $\pi$, as required by Theorem~\ref{thm:robustSafety} and the subsequent Theorem \ref{thm:safeControl}, is designed in Section~\ref {section:stateEstimatorDesign}. A coupled analysis of the observer and the controller is needed to select an appropriate bound $\varepsilon$, see Section~\ref{section:stabilityAnalysis}.
\end{remark}
The following assumptions facilitate the development of a controller satisfying \eqref{eq:controlCondition}. Assumption~\ref{ass:boundedGradientF} and Assumption~\ref{ass:boundedGradientG} are similar to assumptions made in \cite{SCC.Agrawal.Panagou.ea2022}.
\begin{assum}\label{ass:boundedGradientF}
There exists a known Lipschitz continuous function $F: \mathbb{R}^{n} \rightarrow \mathbb{R}$ such that, for all $\hat{x} \in \mathcal{S}$,
\begin{equation}\label{eq:boundedGradientF}
   F(\hat{x}) \leq \inf_{x \in \bar{B}(\hat{x}, \varepsilon)} \nabla_{x} h(x) f(x) + \alpha(h(x)).
\end{equation}
\end{assum}
Assumption~\ref{ass:boundedGradientF} is satisfied, for example, if the functions $x \mapsto \nabla_{x} h(x)f(x)$ and $x \mapsto (\alpha \circ h)(x)$ are locally Lipschitz continuous on $\mathcal{S}$ and their Lipschitz constants are known. The notation $\circ$ represents the function composition operator.
\begin{assum}\label{ass:boundedGradientG}
    There exist known locally Lipschitz continuous functions $ G_{i}^{-}, G_{i}^{+}: \mathbb{R}^{n} \to \mathbb{R}$ for $i = {1, \hdots, m}$, such that 
    \begin{equation}
        G_{i}^{-}(\hat{x}) \leq \left[\nabla_{x}h(x)g(x)\right]_{i} \leq G_{i}^{+}(\hat{x}),
    \end{equation}
    with $\nabla_{x}h(x)g(x) \neq 0$ and for all  $x \in \mathcal{S}$, $\hat{x} \in \{\hat{x}: x \in \bar{B}(\hat{x}, \varepsilon)$\}. Additionally, for all $\hat{x} \in \mathcal{S}$, $\sign( G_{i}^{-}(\hat{x})) = \sign( G_{i}^{+}(\hat{x}))$.
\end{assum}
Assuming $\sign( G_{i}^{-}(\hat{x})) = \sign( G_{i}^{+}(\hat{x}))$, is needed to know whether a positive or negative $u_{i}$ increases $\dot{h}(x, u)$. This assumption facilitates the design of the approximate feedback control policy presented in the following theorem. The theorem is adapted from \cite{SCC.Agrawal.Panagou.ea2022} and a proof is included for completeness.
\begin{thm}\label{thm:safeControl} \cite{SCC.Agrawal.Panagou.ea2022}
    Consider a set $\mathcal{S}$ with a CBF $h: \mathbb{R}^{n} \to \mathbb{R}$ for the system in \eqref{eq:dynamics}, a state observer that ensures that for all $x_{0} \in (\bar{B}(\hat{x}_0,\varepsilon), \mathcal{S})$, the resulting trajectory $x$ satisfies $x(t) \in \bar{B}(\hat{x}(t), \varepsilon)$ for all $t \in \mathbb{R}_{\geq 0}$, and a desired controller $\pi_{\operatorname{des}}: \mathbb{R}^{n} \times \mathbb{R}_{\geq 0} \rightarrow \mathbb{R}^{m}$, which is locally Lipschitz continuous with respect to $\hat{x}$ and piecewise continuous with respect to $t$. Provided Assumptions~\ref{ass:boundedfunctions}--\ref{ass:boundedGradientG} are met, and if the conditions of Theorem~\ref{thm:robustSafety} are satisfied, then the approximate feedback controller $\pi: \mathbb{R}^{n} \times \mathbb{R}_{\geq 0} \to \mathbb{R}^{m}$, determined by solving the quadratic program (QP)
\begin{equation}\label{eq:qp}
\begin{aligned}
\pi(\hat{x}, t) \coloneqq \ &\mathrm{\arg}\min_{u \in \mathbb{R}^{m}} \  \frac{1}{2}\left\|u - \pi_{\operatorname{des}}(\hat{x}, t)\right\|^{2}, \\
\textrm{s.t.} &\ F(\hat{x}) + \sum_{i =  1}^{m}\min\left\{G_{i}^{-}(\hat{x})u_{i}, G_{i}^{+}(\hat{x})u_{i}\right\} \geq 0
\end{aligned} 
\end{equation}
 is locally Lipschitz continuous with respect to $\hat{x}$ and piecewise continuous with respect to $t$.
Furthermore, the closed loop system in \eqref{eq:closedSystem} is safe with respect to the sets $(\bar{B}(\hat{x}_{0}, \varepsilon), \mathcal{S})$.
\end{thm}
\begin{proof}To prove the existence and uniqueness of solutions of the QP, rewrite $\min\{G_{i}^{-}(\hat{x})u_{i}, G_{i}^{+}(\hat{x})u_{i}\}$ in \eqref{eq:qp} using auxiliary variables $z_{i}$ such that
    $z_{i} \leq G_{i}^{-}(\hat{x}) u_{i}$,  $z_{i} \leq G_{i}^{+}(\hat{x}) u_{i}$ for all $i =\{1, \hdots, m\}$. Using the auxiliary variable $z$, the QP in \eqref{eq:qp} can be expressed in standard form as
    \begin{equation*}
    \begin{aligned}
        & \min_{u, z \in \mathbb{R}^{m}} \ \frac{1}{2} u^\top u - \pi_{\operatorname{des}}(\hat{x}, t)^{\top} u,\\
     &\textrm{s.t.} \ 
        \begin{cases}
            F(\hat{x}) + \sum_{i=1}^{m} z_{i} &\geq 0, \\
            z_{i} - G_{i}^{-}(\hat{x}) u_{i} &\leq 0, \\
            z_{i} - G_{i}^{+}(\hat{x}) u_{i} &\leq 0.
        \end{cases}
         \end{aligned}
    \end{equation*}
Stacking the constraints in the QP above and expressing them in matrix form yields the following 
\begin{equation}
    \scalebox{0.9}{$\left[\begin{array}{ccc|ccc}
G_{1}^{-} & \cdots & 0 & -1 & \cdots & 0 \\
G_{1}^{+} & \cdots & 0 & -1 & \cdots & 0 \\
\vdots & \vdots & \ddots & \vdots & \vdots \\
0 & \cdots & G_m^{-} & 0 & \cdots & -1 \\
0 & \cdots & G_m^{+} & 0 & \cdots & -1 \\
\hline
0 & \cdots & 0 & 1 & \cdots & 1
\end{array}\right]
\left[\begin{array}{c}
u_{1} \\ \vdots \\ u_m \\
\hline
z_{1} \\ \vdots \\ z_m
\end{array}\right]
\geq
\left[\begin{array}{c}
0 \\
0 \\
\vdots \\
0 \\
0 \\
\hline
-F
\end{array}\right]$},
\end{equation}
where the dependence on $\hat{x}$ is omitted for brevity. The resulting constraint matrix has size $(2m + 1) \times 2m$. However, under Assumption~\ref{ass:boundedGradientG}, for each $i$, only one of either the $(2i-1)\text{-th}$ or the $(2i)\text{-th}$ constraint involving $z_{i}$ can be active since $\sign(G_{i}^{-}) = \sign(G_{i}^{+})$ and $\nabla_{x} h(x) \neq 0$ (implying $G_{i}^{-}, G_{+} \neq 0$). In particular, if $G_{i}^{-} \neq G_{i}^{+}$, then the two constraints cannot be simultaneously active, and if $G_{i}^{-} = G_{i}^{+} \neq 0$, then both constraints are equivalent, and either one can be active. Hence, for each $i$, at most one of the two constraints involving $z_{i}$ can be active, ensuring that the QP remains feasible. Consider the submatrix corresponding to the active constraints. This submatrix will have at most $m + 1$ non-zero rows with each row corresponding to an active constraint. Therefore, the submatrix has $m+1$ linearly independent rows and must have full rank, implying the QP with $2m$ decision variables has a non-empty set of feasible solutions. Furthermore, since the quadratic cost function is strictly convex, there exists a unique minimizer. 

Local Lipschitz continuity of $\hat{x} \mapsto \pi(\hat{x}, t)$ follows since the linear independence of the rows of the active constraint matrix satisfies the regularity conditions in \cite[Theorem~3.1]{SCC.Hager.ea1979}, and the functions $(\hat{x}, t) \mapsto \pi_{\operatorname{des}}(\hat{x}, t)$, $\hat{x} \mapsto F(\hat{x})$, $\hat{x} \mapsto G_{i}^{-}(\hat{x})$, and $ \hat{x} \mapsto G_{i}^{+}(\hat{x})$ are locally Lipschitz continuous on $\mathcal{S}$.

Finally, since $\nabla_{x}h(x)g(x)u = \sum_{i = 1}^{m}\left[\nabla_{x}h(x)g(x)\right]_{i} u_{i} \geq \sum_{i=1}^{m}\min\{G_{i}^{-}(\hat{x})u_{i}, G_{i}^{+}(\hat{x})u_{i}\}$, satisfying the QP constraints also ensures \eqref{eq:controlCondition}, then safety follows by Theorem~\ref{thm:robustSafety}.
\end{proof}

Since the constraint in \eqref{eq:qp} is non-smooth, a closed-form analytical solution to the QP cannot be obtained using standard optimization techniques like the KKT optimality conditions \cite{SCC.Bryson.Ho1975}. As a result, a numerical solver is employed to solve the QP in \eqref{eq:qp} and to satisfy the constraints. 

\section{Stabilizing Optimal Control Design}
\label{section:stabilizingControlDesign}

In this section, the desired policy $\pi_{\operatorname{des}}$ in \eqref{eq:qp} is designed by solving an unconstrained version of the optimal control problem in Problem~\ref{prob:coop}, inspired by results from \cite{SCC.Kamalapurkar.Walters.ea2016, SCC.Modares.Lewis2014}, and \cite{SCC.Ogri.Mahmud.ea2023}. The objective is to derive an optimal feedback policy $u^{*}: \mathbb{R}^{n} \to \mathbb{R}^{m}$ that, when applied to the system in \eqref{eq:dynamics}, minimizes the cost functional in \eqref{eq:costFunctional} over the resulting closed-loop trajectories. Assuming the optimal policy exists, the optimal value function $V^{*}: \mathbb{R}^{n}  \to \mathbb{R}$, which is the cost-to-go obtained by starting from initial state $x_{0}$ and following the control policy $u^{*}$, can be expressed as  
\begin{equation}\label{eq:valuefunction}
    V^{*}(x_{0}) \coloneqq \min_{u:\mathbb{R}_{\geq t} \to \mathbb{R}^{m}}\int_{t}^\infty r\left(x \left(\tau\right), u \left(\tau\right)\right)  \mathrm{d}\tau.
\end{equation}
 A general solution to the optimal control problem in \eqref{eq:costFunctional} can be obtained by solving the corresponding Hamilton-Jacobi-Bellman (HJB) equation
\begin{equation}\label{eq:HJB} 
\min_{u\in \mathbb{R}^{m}} \left(\nabla_{x} V^{*}(x)\left(f(x)+g(x)u \right) + Q(x) + u^\top R u\right) = 0,
\end{equation}
with the boundary condition $V^{*}(0) = 0$.
 According to \cite[Theorem~1.5]{SCC.Kamalapurkar.Walters.ea2018}, provided the optimal value function $V^{*}$ is continuously differentiable, it is the unique positive definite solution of the HJB equation. The stabilizing optimal control policy $u^{*}$ can be obtained by solving \eqref{eq:HJB} as
\begin{equation}\label{eq:optimalcontrol}
u^{*}(x) \coloneqq -\frac{1}{2}R^{-1}g(x)^{\top}\left(\nabla_{x}V^{*}(x)\right)^{\top}.
\end{equation}
Computation of $V^{*}$ requires solution of the closed-loop HJB equation 
\begin{multline}\label{eq:optimalHamiltonian}
    -\frac{1}{4}\nabla_{x} V^{*}(x)g(x)R^{-1}g^{\top}(x)\left(\nabla_{x}V^{*}(x)\right)^{\top}\\ + \nabla_{x} V^{*}(x)f(x) + Q(x) = 0, \quad  \forall x \in \mathbb{R}^{n},
\end{multline}
which is difficult to obtain analytically for a general class of nonlinear systems. The following section introduces an approximate form of the optimal control policy in \eqref{eq:optimalcontrol}.

\subsection{Approximate desired control policy}
  To facilitate the approximation of the value function, define the compact set $\Omega \subset \mathbb{R}^{n}$ containing the origin such that $\mathcal{S} \subseteq \Omega$. Since the value function $V^{*}$ is continuous, it can be represented using the universal approximation property of neural networks as
\begin{equation} \label{eq:optimalV}
V^{*}(x) = W^{\top} \sigma(x)+\epsilon(x), \quad \forall x \in \Omega,
\end{equation}
where $W\in\mathbb{R}^{L}$ is an unknown ideal vector of bounded weights, $\sigma:\mathbb{R}^{n}\to\mathbb{R}^{L}$ is a vector of continuously differentiable nonlinear activation functions such that $\sigma\left(0\right)=0$, $L\in\mathbb{N}$ is the number of basis functions, and $\epsilon:\mathbb{R}^{n}\to\mathbb{R}$ is the reconstruction error. Applying the Universal Approximation Theorem \cite[Theorem 1.5]{SCC.Sauvigny2012}, given $\overline{\epsilon} > 0$, the activation functions $\sigma$ can be selected so that the weights and the approximation errors satisfy $\sup_{x  \in \Omega}\|W\| \leq \overline{W}$, $ \sup_{x  \in \Omega}\|\sigma(\cdot)\| \leq \overline{\sigma}$, $ \sup_{x  \in \Omega}\|\nabla_{x}\sigma(\cdot)\| \leq \overline{\nabla\sigma}, \sup_{x  \in \Omega}\|\epsilon(\cdot)\| \leq \overline{\epsilon}$, and $\sup_{x  \in \Omega}\|\nabla_{x}\epsilon(\cdot)\| \leq \overline{\epsilon}$, for some positive constants $\overline{W}$ and $\overline{\sigma} \in \mathbb{R}_{> 0}$. Employing the actor-critic RL approach, actor and critic weights $\hat{W}_{a} \in\mathbb{R}^{L}$ and $\hat{W}_{c} \in\mathbb{R}^{L}$ can be used to estimate the unknown ideal weights $W$ \cite{SCC.Kamalapurkar.Walters.ea2018}. Using the critic weights, the approximate value function, denoted by $\hat{V}:\mathbb{R}^{n}\times\mathbb{R}^{L}\to\mathbb{R}$, is defined as
\begin{equation}\label{eq:VApprox}
\hat{V}(\hat{x},\hat{W}_{c})\coloneqq \hat{W}_{c}^{\top} \sigma(\hat{x}),
\end{equation}
and using the actor weights, the approximate control policy, denoted by $\hat{u}:\mathbb{R}^{n} \times \mathbb{R}^{L} \to \mathbb{R}^{m}$, is defined as
\begin{equation}\label{eq:uApprox}
\hat{u}(\hat{x}, \hat{W}_{a}) \coloneqq -\frac{1}{2}R^{-1}g(\hat{x})^{\top}\nabla_{\hat{x}}\sigma(\hat{x})^{\top}\hat{W}_{a}.
\end{equation}

\subsection{Bellman Error and Simulation of Experience}
Substituting \eqref{eq:VApprox} and \eqref{eq:uApprox} into \eqref{eq:HJB} results in a residual term
 $\hat{\delta}: \mathbb{R}^{n} \times \mathbb{R}^{L} \times \mathbb{R}^{L} \to \mathbb{R}$, commonly referred to as the Bellman error (BE), defined as 
\begin{multline} \label{eq:BE1}
\hat{\delta}(\hat{x}, \hat{W}_{c}, \hat{W}_{a}) \coloneqq   \nabla_{\hat{x}}\hat{V}(\hat{x},\hat{W}_{c})\left(f(\hat{x}) + g(\hat{x})\hat{u}(\hat{x}, \hat{W}_{a})\right) \\
 + Q(\hat{x}) + \hat{u}(\hat{x}, \hat{W}_{a})^{\top} R \hat{u}(\hat{x}, \hat{W}_{a}).
\end{multline}
The control objective is achieved by simultaneously adjusting the actor and critic weights, $\hat{W}_{a}$ and $\hat{W}_{c}$, to minimize the BE.

To estimate the value function, online RL methods require persistence of
excitation (PE) \cite{SCC.Modares.Lewis.ea2013, SCC.Kamalapurkar.Rosenfeld.ea2016}, which is difficult to guarantee in practice. However, through BE extrapolation, stability and convergence of online RL can be established using Assumption~\ref{ass:rankCond} \cite{SCC.Kamalapurkar.Rosenfeld.ea2016}.
To that end, we select a set of trajectories $\left\{ x_{k}: \mathbb{R}_{\geq 0} \to \mathbb{R}^{n} \mid k=1,\cdots, N\right\}$ and extrapolate the BE along these trajectories, where $N \in \mathbb{N}$ denotes the total number of extrapolation points.
Let the actor and critic weight estimation errors be defined as $\tilde{W}_{a} \coloneqq W -\hat{W}_{a}$ and $\tilde{W}_{c}\coloneqq W -\hat{W}_{c}$, respectively. To facilitate the subsequent stability analysis in Section~\ref{section:stabilityAnalysis}, the extrapolated BEs can be expressed in terms of the weight estimation errors $\tilde{W}_{a}$ and $\tilde{W}_{c}$ as
\begin{equation}
    \hat{\delta}_{k} = -\omega_{k}^{\top} \tilde{W}_{c} + \frac{1}{4}\tilde{W}_{a}^{\top}G_{\sigma}\tilde{W}_{a} + \Delta_{k},
\end{equation}
where $\hat{\delta}_{k} \coloneqq  \hat{\delta}(x_{k}, \hat{W}_{c}, \hat{W}_{a})$, $\omega_{k} \coloneqq \sigma_{k}(f_{k}+g_{k}\hat{u}(x_{k},\hat{W}_{a}))$, $\Delta_{k} \coloneqq \frac{1}{2}W^{\top} \nabla_{x_{k}} \sigma_{k} G_{R_{k}} \nabla_{x_{k}} \epsilon_{k}^{\top}+\frac{1}{4}G_{\epsilon_{k}}-\nabla_{x_{k}}\epsilon_{k} f_{k}$, $G_{\epsilon_{k}}\coloneqq \nabla_{x_{k}} \epsilon_{k} G_{R_{k}} \nabla_{x_{k}} \epsilon_{k}^{\top}$, $G_{\sigma_{k}} \coloneqq \nabla_{x_{k}} \sigma_{k} G_{R_{k}} \nabla_{x_{k}} \sigma_{k}^{\top}$, $G_{R_{k}}\coloneqq g_{k}R^{-1}g_{k}^{\top}$, $f_{k} \coloneqq f(
x_{k})$, $g_{k} \coloneqq g(
x_{k})$, $\sigma_{k} \coloneqq \sigma (x_{k})$, and $\epsilon_{k} \coloneqq \epsilon(x_{k})$. Since the ideal weights, the reconstruction error, and the gradient of the reconstruction error are bounded on the compact set $\Omega$, and since $x_{k} \mapsto f(x_{k})$ and $x_{k} \mapsto g(x_{k})$ are locally Lipschitz continuous $x_{k}$ on $\Omega$, there exists a constant $\overline{\Delta}_{k} > 0$ such that $\sup_{x_{k} \in \Omega} |\Delta_{k}| \leq \overline{\Delta}_{k}$ provided the extrapolation trajectories are selected such that $x_{k}(t) \in \Omega$ for all $t\geq 0$ and $k = 1, \hdots, N$.

\subsection{Update laws for Actor and Critic weights}
 The weights are updated using the extrapolated BEs $\hat{\delta}_{k}$ as
\begin{align}
    &\dot{\hat{W}}_{c} =- \frac{k_{c}}{N}\Gamma\sum_{k=1}^{N}\frac{\omega_{k}}{\rho_{k}}\hat\delta_{k},\label{eq:criticUpdate}\\
    &\dot{\Gamma} = \beta\Gamma- \frac{k_{c}}{N}\Gamma\sum_{k=1}^{N}\frac{\omega_{k}\omega_{k}^{\top}}{\rho_{k}^{2}}\Gamma,\label{eq:gammaUpdate}\\
    &\dot{\hat{W}}_{a} = \proj_{\bar{B}(0,\overline{W})}\Bigl(f_{a}(\hat{W}_{a}, \hat{W}_{c})\Bigr),\label{eq:actorUpdate}
\end{align}
with $\Gamma\left(0\right)=\Gamma_{0}$, where $f_a(\hat{W}_a, \hat{W}_c) \coloneqq -k_{a_{1}}(\hat{W}_{a}-\hat{W}_{c})\!+\sum_{k=1}^{N}\frac{k_{c}G_{\sigma_{k}}^{\top}\hat{W}_{a}\omega_{k}^{\top}}{4N\rho_{k}}\hat{W}_{c}-k_{a_{2}}\hat{W}_{a}$, $\Gamma:\mathbb{R}_{\geq 0} \to \mathbb{R}^{L\times L}$
is a time-varying least-squares gain matrix, $\rho_{k}\left(t\right)\coloneqq 1+\nu\omega_{k}^{\top}\left(t\right)\omega_{k}\left(t\right)$, $\nu > 0$ is a constant positive normalization gain, $\beta > 0$ is a constant forgetting factor, $k_{c},k_{a_{1}},k_{a_{2}} > 0$ are constant adaptation gains, and $\proj(\cdot)$ is the smooth projection operator defined in \cite{SCC.Cai.Queiroz.ea2006} and \cite{SCC.Krstic.Kanellakopoulos.ea1995}, employed to ensure that $\|\tilde{W}_{a}(t)\| \leq \overline{W}$ for all $t$. The following rank condition must be satisfied to guarantee convergence of the closed-loop system.
\begin{assum}
    \label{ass:rankCond}There exists a constant $\underline{c}_{1}$ such that the set of extrapolation trajectories $\left\{x_{k}: \mathbb{R}_{\geq 0} \mid k=1,\hdots,N\right\}$ satisfies
    \begin{equation}
    0 < \underline{c}_{1} \leq \inf_{t\in \mathbb{R}_{\geq 0}}\left\{\lambda_{\min}\left(\frac{1}{N}\sum_{k=1}^{N}\frac{\omega_{k}(t)\omega_{k}^{\top} (t)}{\rho_{k}^{2}(t)}\right)\right\}.\label{eq:CLBCPE2}
    \end{equation}
\end{assum}
As described in \cite{SCC.Mahmud.Nivison.ea2021}, since $\omega_{k}$ is a function of $x_{k}$ and $\hat{W}_{c}$,  Assumption \ref{ass:rankCond} cannot be verified \emph{a priori}. However, unlike the PE condition utilized in \cite{SCC.Vamvoudakis.Lewis2010}, Assumption \ref{ass:rankCond} can be monitored online. Furthermore, since $\lambda_{\min}\left\{\sum_{k=1}^{N}\frac{\omega_{k}(t)\omega_{k}^{\top} (t)}{\rho_{k}^{2}(t)}\right\}$ is non-decreasing in the number of samples, $N$, Assumption \ref{ass:rankCond} can be met, heuristically, by increasing the number of 
extrapolation trajectories. The calculation of a precise bound on the number of extrapolation trajectories needed for the subsequent stability analysis is out of the scope of this paper.

The desired locally Lipschitz continuous controller $(\hat{x}, t) \mapsto \pi_{\operatorname{des}}(\hat{x}, t)$ in \eqref{eq:qp} is now defined using the actor weights and the state estimates as 
\begin{equation}\label{eq:uControl}
    \pi_{\operatorname{des}}(\hat{x}, t) \coloneqq \hat{u}\left(\hat{x}, \hat{W}_{a}(t)\right).
\end{equation}
Since an explicit closed-form solution to the QP in \eqref{eq:qp} cannot be derived, the following lemma is necessary to facilitate the subsequent stability analysis in Section~\ref{section:stabilityAnalysis}. 
\begin{lem}
There exists a constant $\epsilon_{\pi} \in \mathbb{R}_{\geq 0}$ such that
\begin{equation}\label{eq:controlBound}
\left\| \pi(\hat{x}, t) - \pi_{\operatorname{des}}(\hat{x}, t)\right\| \leq \epsilon_{\pi}, \quad \forall \hat{x} \in \Omega, \forall t \in \mathbb{R}_{\geq 0}.
\end{equation}
\end{lem}
\begin{proof}
The projection of actor weights $\hat{W}_a$ in \eqref{eq:actorUpdate} ensures $(\hat{x}, t) \mapsto \pi_{\operatorname{des}}(\hat{x}, t)$ is bounded for all $t$ provided $\hat{x} \in \Omega$. Let $\pi(\hat{x}, 0)$ be the safe control policy in \eqref{eq:qp} at initial time $t=0$. Since the constraint in \eqref{eq:qp} is time-independent, $\pi(\hat{x}, 0)$ remains feasible for all $t \geq 0$. Furthermore, since $\pi(\hat{x}, t)$ is the optimal solution to \eqref{eq:qp}, by optimality, 
\begin{equation*}
\|\pi(\hat{x}, t) - \pi_{\operatorname{des}}(\hat{x}, t)\| \leq \|\pi(\hat{x}, 0) - \pi_{\operatorname{des}}(\hat{x}, t)\|.
\end{equation*}
Applying the triangle inequality and given that $\pi(\hat{x}, 0)$ and $\pi_{\operatorname{des}}(\hat{x}, t)$ are both bounded by a time-independent bound as long as $\hat{x} \in \Omega$, $\exists \epsilon_{\pi} > 0$, independent of $t$, such that
\begin{equation}
\|\pi(\hat{x}, t) - \pi_{\operatorname{des}}(\hat{x}, t)\| \leq \|\pi(\hat{x}, 0)\| + \|\pi_{\operatorname{des}}(\hat{x}, t)\| \leq \epsilon_{\pi},
\end{equation}
$\forall \hat{x} \in \Omega$ and $\forall t \in \mathbb{R}_{\geq 0}$.
\end{proof}

In the next section, we introduce a \emph{bounded-error} state observer that satisfies $x(t) \in \bar{B}(\hat{x}(t), \varepsilon)$ for all $t \in \mathbb{R}_{\geq 0}$, as required by Theorems~\ref{thm:robustSafety} and \ref{thm:safeControl}, under the control law in \eqref{eq:qp}.

\section{DNN State Observer Design}\label{section:stateEstimatorDesign}

To facilitate the development of the observer, let $A \in \mathbb{R}^{n \times n}$ be a Hurwitz matrix such that the pair $(A, C)$ is observable. Adding and subtracting the term $Ax$ from the system in \eqref{eq:dynamics} yields
 \begin{equation}
 \begin{aligned}
  \dot{x} &= Ax +f_{0}(x)+g(x)u\\
y &= Cx, 
\end{aligned}\label{eq:dynamics_new}
\end{equation}
 where $f_{0}(x) \coloneqq f(x)-Ax$. Since $f$ is unknown, so is $f_{0}$. Motivated by the DNN-based observers in \cite{ SCC.Yang.Liu.Huang.ea2013, SCC.Greene.Bell.ea2023, SCC.Zegers.Sun.ea2023}, the unknown continuous function $f_{0}$ can be approximated on the compact set $\Omega$ using a DNN as 
\begin{equation}
f_{0}(x, t) = \theta^{\top}\phi\left(\Phi(x, t)\right) + \epsilon_{\theta}(x), \quad \forall x\in \Omega, \label{eq:fDNNdynamics}
\end{equation}
where $\epsilon_{\theta}: \mathbb{R}^{n} \to \mathbb{R}^{n}$ is the function approximation error, $\theta \in \mathbb{R}^{p \times n}$ is an unknown bounded ideal output-layer DNN weight matrix, $\phi: \mathbb{R}^{h} \times \mathbb{R}_{\geq 0} \to \mathbb{R}^{p}$ is a vector of continuously differentiable activation functions, $\Phi: \mathbb{R}^{n} \to \mathbb{R}^{h}$ represents the ideal inner layer features of the DNN at time $t$ defined as 
\begin{equation}
    \Phi(x, t) \coloneqq W_{k}^{\top}(t)\varphi_{k}\circ W_{k-1}^{\top}(t)\varphi_{k-1} \circ \hdots \circ  \ W_{1}^{\top}(t)\varphi_{1}\circ x,
\end{equation} for all $x \in \Omega$, where $W_{k}(t)$ and $\varphi_{k}(\cdot)$ denote the weights and activation functions of the $k\text{-th}$ inner layer of the DNN at time $t$, respectively, and $k \in \mathbb{N}$ is the total number of inner layers of the
DNN. Since the inner layer weights are updated intermittently, each weight $W_{k}(t)$ is piecewise constant in time, i.e., $W_{k}(t) = W_{k, j}$, for $ t \in [t_{j}, t_{j+1})$, where $\{t_{j}\}_{j\in \mathbb{N}}$ denote the update times. According to the Universal Approximation Theorem \cite[Theorem 1.5]{SCC.Sauvigny2012}, the activation function $\phi$ can be selected so that given any compact set $\Omega$, the ideal DNN weights, the activation functions, the approximation error and their gradients satisfy $\|\theta\| \leq \overline{\theta} < \infty$,  $\sup_{x\in \Omega}\|\Phi(x, t)\| \leq \overline{\phi}$,  $ \sup_{x\in \Omega}\|\nabla_{(x)}\Phi(x, t)\| \leq \overline{\nabla\phi}$, $\sup_{x\in \Omega}\|\epsilon_{\theta}(x)\| \leq\overline{\epsilon_{\theta}}$, $\sup_{x\in \Omega}\|\nabla_{(x)}\epsilon_{\theta}(x)\| \leq\overline{\epsilon}_{\theta}$ for some constants $\overline{\phi}$, $\overline{\nabla\phi}$, and $\overline{\epsilon_{\theta}} \in \mathbb{R}_{>0}$.
From \eqref{eq:fDNNdynamics}, the dynamics in \eqref{eq:dynamics_new} can be reformulated into an augmented system given as
\begin{equation}
  \dot{x} = Ax + \theta^{\top}\phi\left(\Phi(x, t)\right)  + \epsilon_{\theta}(x)  + g(x)u. \label{eq:dynamicsDNN}
\end{equation}

The estimate of the DNN representation of $f_{0}$ is defined as
\begin{equation}
    \hat{f}_{0}(\hat{x}, \hat{\theta}, t) \coloneqq \hat{\theta}^{\top}\phi\left(\Phi(\hat{x}, t)\right),\label{eq:fEstimateDNNdynamics}
\end{equation} where $\hat{\theta} \in \mathbb{R}^{p \times n}$ is an estimate of the ideal DNN outer layer weight matrix and 
\begin{equation}
    \Phi(\hat{x}, t) \coloneqq W_{k}^{\top}(t)\varphi_{k}\circ W_{k-1}^{\top}(t)\varphi_{k-1}\circ\hdots\circ W_{1}^{\top}(t)\varphi_{1}\circ\hat{x}
\end{equation} is the estimate of the inner layer features.

To estimate the states of the system in \eqref{eq:dynamics}, a Luenberger-like state observer is designed using the DNN-based estimate in \eqref{eq:fEstimateDNNdynamics} and the measured output $y$ as
\begin{equation}\label{eq:DNNDynamicsEstimate}
    \dot{\hat{x}} =  A\hat{x}+ \hat{\theta}^{\top}\phi(\Phi(\hat{x}, t)) + g(\hat{x})u + K\left(y - C\hat{x}\right),
\end{equation}
with $\hat{x}_{0} = \hat{x}(0)$, where $K \in \mathbb{R}^{n \times q}$ is the Luenberger observer gain which is selected such that $( A - KC)$ is a Hurwitz matrix satisfying the Lyapunov equation $(A - KC)^{\top}P + P(A - KC) = -S$ for a given $S = S^{\top} \succ 0$ and some $P = P^{\top} \succ 0$.

Let the state estimation error be defined as $\tilde{x}\coloneqq x - \hat{x}$ with the state estimation error dynamics are given by
\begin{equation}\label{eq:DNNObervererrorDyn}
    \dot{\tilde{x}} = (A-KC)\tilde{x}+  \tilde{\theta}^{\top}\phi\left(\Phi(\hat{x}, t)\right) + e(x, \hat{x}, u, t),
\end{equation}
where $\tilde{\theta} \coloneqq \theta  - \hat{\theta} \in \mathbb{R}^{p \times n}$ is the outer layer weight estimation error and $e(x, \hat{x}, u, t) \coloneqq \theta^{\top}\left(\phi(\Phi(x, t)) - \phi(\Phi(\hat{x}, t))\right) + \left(g(x)-g(\hat{x})\right)u + \epsilon_{\theta}(x)$. The stability analysis in Section~\ref{section:stabilityAnalysis} relies on the following assumption.
\begin{assum}\label{ass:boundedActivation}
    The user-selected inner layer features $(x, t) \mapsto \Phi(x, t)$ are locally Lipschitz continuous in $x$ over $\Omega$, uniformly in $t$. That is, for all $x, \hat{x} \in \Omega$ and all $t \in \mathbb{R}_{\geq 0}$, there exists a constant $L_{\Phi} > 0$ such that
\begin{equation}
\|\Phi(x, t) - \Phi(\hat{x}, t)\| \leq L_{\Phi}\|x - \hat{x}\|.
\end{equation}
\end{assum}

Under Assumption~\ref{ass:boundedActivation} and since $\phi$ is continuously differentiable, the Mean Value Theorem can be invoked to show that $\|\phi\left(\Phi(x, t)\right)-\phi\left(\Phi(\hat{x}, t)\right)\| \leq  \overline{\nabla\phi}L_{\Phi}\|x -\hat{x}\|$,  $\forall x, \hat{x} \in \Omega$, and $\forall t \in \mathbb{R}_{\geq 0}$. The following Lemma is necessary to facilitate the development of an ICL-based update law for the outer layer DNN weight estimates.

\begin{lem}\label{lem:ErrorTermformulation}
Given time delay $\Delta t > 0$, the state trajectory satisfies,
\begin{equation}\label{eq:StateEstimateRelation}
\hat{\mathcal{X}}(t) = \theta^{\top} \mathcal{Y}(t) + \mathcal{G}_{u}(t) + \mathcal{E}(t), \quad \forall t \geq \Delta t,
\end{equation}
where $\hat{\mathcal{X}}(t) = \hat{x}(t) -\hat{x}\left(t-\Delta t\right)$, $\mathcal{Y}(t) \coloneqq \int_{t-\Delta t}^{t} \phi\left(\Phi(\hat{x}(\tau), \tau)\right)\mathrm{d}\tau$, $\mathcal{G}_{u}(t) \coloneqq \int_{t-\Delta t}^{t} A\hat{x}(\tau) + g(\hat{x}(\tau))u(\tau)\ \mathrm{d}\tau$, and $\mathcal{E}(t) \coloneqq  -\tilde{x}(t)+\tilde{x}(t-\Delta t)  + \int_{t-\Delta t}^{t} A\tilde{x}(\tau) + e\left(x(\tau), \hat{x}(\tau), u(\tau), \tau\right) \mathrm{d}\tau$.
\end{lem}
\begin{proof}
     Follows from the Fundamental Theorem of Calculus.
\end{proof}
The residual term $\mathcal{E}$ has a norm that is small
enough in a sense made precise in the subsequent stability analysis in Section~\ref{section:stabilityAnalysis}.
Lemma~\ref{lem:ErrorTermformulation} implies that the outer layer DNN weight error at any time $t \in \mathbb{R}_{\geq 0}$ can be expressed as 
\begin{equation}\label{eq:linearMeasurement}
    \tilde{\theta}(t)^{\top} \mathcal{Y}(t) + \mathcal{E}(t) = \hat{\mathcal{X}}(t) - \hat{\theta}(\tau)^{\top} \mathcal{Y}(t) - \mathcal{G}_{u}(t). 
\end{equation}
Although the left-hand side of \eqref{eq:linearMeasurement} is affine due to the residual term $\mathcal{E}(t)$, it will be treated as a linear measurement model to develop the update law for the outer-layer weight estimates. Specifically, we will use stored values of the right-hand side of \eqref{eq:linearMeasurement} to update $\hat{\theta}$. Since the residual term $\mathcal{E}(t)$ depends on the state estimation error $\tilde{x}(t)$, it is likely to be large during the transient period where $\tilde{x}(t)$ is large due to the lack of full state feedback, leading to inaccurate estimates of the outer layer weights. Since the observer in \eqref{eq:DNNDynamicsEstimate} ensures that $\tilde{x}(t)$ decays exponentially to a small ball, as proven in the stability analysis in Section~\ref{section:stabilityAnalysis}, the residual term can be made smaller by periodically purging erroneous data from the history stack as described in Algorithm~\ref{algo:parameterEstimatorAlgo}.

At any given time, two separate history stacks, $\mathcal{H}$ and $\mathcal{M}$, are maintained as shown in Algorithm~\ref{algo:parameterEstimatorAlgo}, where $\mathcal{H}$ is the active history stack and $\mathcal{M}$ is the auxiliary history stack. New data are added to the auxiliary history stack. The residual error is decreased by implemented by periodically replacing $\mathcal{H}$ with $\mathcal{M}$ and emptying $\mathcal{M}$. Let $\varrho : \mathbb{R}_{\geq 0} \to \mathbb{N}$ be a switching signal with $\varrho(0) = 1$. For any $t \geq 0$, $\varrho(t) = j + 1$, where $j$ denotes the number of times the replacement $\mathcal{H} \gets \mathcal{M}$ has occurred over the interval $[0, t]$.
For each $s \in \mathbb{N}$, we define the time interval
$\mathcal{I}_{s} = [T_{s-1},T_s) \coloneqq \{t \in \mathbb{R}_{\geq 0} \mid \varrho(t) = s\}$, where $T_s$ denotes the time when the $(s+1)-$th subsystem is switched on, i.e., when $\varrho$ changes from $s$ to $s+1$, with $T_0 \coloneqq 0$.

The history stack that is active during $\mathcal{I}_{s}$ is denoted by $\mathcal{H}_{s}$ and defined as
\begin{equation}\label{eq:HS}
\mathcal{H}_{s}\coloneqq \left\{ ( \mathcal{Y}_{s_{i}},\, \hat{\mathcal{X}}_{s_{i}},\, \mathcal{G}_{u_{s_{i}}})\right\}_{i=1}^{M},
\end{equation}
where $\hat{\mathcal{X}}_{s_{i}} \coloneqq \hat{\mathcal{X}}(t_{i,s})$, $\mathcal{Y}_{s_{i}} \coloneqq \mathcal{Y}(t_{i,s})$, $\mathcal{G}_{u_{s_{i}}} \coloneqq \mathcal{G}_u(t_{i,s})$, and $\{t_{i,s}\}_{i=1}^M\subset \mathcal{I}_{s-1}$ denote the time at which the $i$-th datum in $\mathcal{H}_{s}$ was recorded. Let $\Sigma_{\mathcal{Y}_{s}}: \mathbb{R}_{\geq 0} \to \mathbb{R}^{p\times p}$ and $\Sigma_{\mathcal{E}_{s}}: \mathbb{R}_{\geq 0} \to \mathbb{R}^{p\times n}$ be defined as $\Sigma_{\mathcal{Y}_{s}} \coloneqq \sum_{i=1}^{M}\frac{\mathcal{Y}_{s_{i}}\mathcal{Y}_{s_{i}}^{\top}}{1+\kappa\|\mathcal{Y}_{s_{i}}\|^{2}} \in \mathbb{R}^{p \times p}$ and $\Sigma_{\mathcal{E}_{s}} \coloneqq \sum_{i=1}^{M}\frac{\mathcal{Y}_{s_{i}}\mathcal{E}_{s_{i}}^{\top}}{1+\kappa\|\mathcal{Y}_{s_{i}}\|^{2}} \in \mathbb{R}^{p \times n}$, where $\mathcal{E}_{s_{i}} \coloneqq \mathcal{E}(t_{i,s})$. Note that $H_s$ is recorded over $\mathcal{I}_{s-1}$ and does not change over $\mathcal{I}_s$.

The auxiliary history stack that is being recorded during $\mathcal{I}_s$ is denoted by $\mathcal{M}_{s}$. Let $(\mathcal{Y}_{*},\, \hat{\mathcal{X}}_{*},\, \mathcal{G}_{u_{*}})$ denote a datum obtained at $t=t^{*}\in \mathcal{I}_s$. If
\begingroup\medmuskip=0mu\begin{multline}\label{eq:eigenValueCondition}
    \scalemath{0.97}{\max_{i=1:M}\left\{\lambda_{\min}\left(\Sigma_{\mathcal{Y}_{s}} - \frac{\mathcal{Y}_{s_{i}}\mathcal{Y}_{s_{i}}^{\top}}{1+\kappa\|\mathcal{Y}_{s_{i}}\|^{2}} + \frac{\mathcal{Y}_{*}\mathcal{Y}_{*}^{\top}}{1+\kappa\|\mathcal{Y}_{*}\|^{2}}\right)\right\}} \\-\lambda_{\min}\left(\Sigma_{\mathcal{Y}_{s}}\right) \geq \lambda^{*}
\end{multline}\endgroup
for some user-selected eigenvalue threshold $\lambda^{*} \in \mathbb{R}_{\geq 0}$, then $(\mathcal{Y}_{*},\, \hat{\mathcal{X}}_{*},\, \mathcal{G}_{u_{*}})$ replaces the datum $( \mathcal{Y}_{s_{j}},\, \hat{\mathcal{X}}_{s_{j}},\, \mathcal{G}_{u_{s_{j}}})$ currently stored in $\mathcal{M}_{s}$, where $j$ denotes the maximizer in \eqref{eq:eigenValueCondition}. If $\mathcal{M}_{s}$ is full
and sufficient dwell time has passed, then $\mathcal{M}_{s}$ replaces the data in $\mathcal{H}_{s}$ (i.e., $\mathcal{H}_{s+1}$ is set to be equal to $\mathcal{M}_{s})$ and $\mathcal{M}_{s}$ is emptied. 

The ICL-based update law for the outer-layer weight estimates is then designed as
\begin{equation}\label{eq:outerLayerWeightsUpdate}
    \dot{\hat{\theta}} = \proj_{\bar{B}(0,\overline{\theta})}\left(f_{\theta}(\hat{\theta}, t)\right),
\end{equation}
with $f_{\theta}(\hat{\theta}, t) \coloneqq k_{\theta}\gamma\sum_{i=1}^{M}\frac{\mathcal{Y}_{\varrho(t)_{i}}\left(\hat{\mathcal{X}}_{\varrho(t)_{i}}-\mathcal{G}_{u_{\varrho(t)_{i}}}-\hat{\theta}^{\top}\mathcal{Y}_{\varrho(t)_{i}}\right)^{\top}}{1+\kappa\|\mathcal{Y}_{\varrho(t)_{i}}\|^{2}}$, where $\kappa, k_{\theta} \in \mathbb{R}_{> 0}$, and $\gamma \in \mathbb{R}^{p \times p}$ are user-defined constant adaptation gains. To ensure that the outer layer DNN weight error satisfies the bound $\|\tilde{\theta}(t)\| \leq \overline{\theta}$ for all $t$, the smooth projection operator $\proj(\cdot)$ introduced in \cite{SCC.Cai.Queiroz.ea2006} and \cite[Appendix E]{ SCC.Krstic.Kanellakopoulos.ea1995} is employed.

Under Lemma \ref{lem:ErrorTermformulation}, the outer layer DNN weight estimation error dynamics can be expressed as 
\begin{equation}\label{eq:weightErrorDyn}
    \dot{\tilde{\theta}} = -\proj_{\bar{B}(0,\overline{\theta})}{\left(k_{\theta}\gamma\Sigma_{\mathcal{Y}_{\varrho(t)}}\tilde{\theta} +k_{\theta}\gamma\Sigma_{\mathcal{E}_{\varrho(t)}}\right)},
\end{equation}
respectively. Since $t \mapsto \varrho(t)$ is piecewise continuous, the trajectories of \eqref{eq:weightErrorDyn} are well-defined in the sense of Carathéodory \cite{SCC.Hale1980}. 
It is clear from \eqref{eq:weightErrorDyn} that for the outer layer DNN weight estimation error to converge, the matrix $\Sigma_{\mathcal{Y}_{s}}$ needs to be positive definite for each $s$. Such a positive definite matrix can be obtained if the trajectories are sufficiently informative and the time instances $\{t_{i,s}\}_{i=1}^{M}\subset \mathcal{I}_{s-1}$ are selected carefully. This requirement is formalized by the following assumption.
\begin{assum}\label{ass:regressorRank}
There exist constants $M \in \mathbb{N}$ and $\underline{\sigma}_{\theta}, T > 0$ such that for all $s \in \mathbb{N}$ there exist time instances $\{t_{i,s}\}_{i=1}^M$ such that $T_{s-1} \leq t_{1,s} < t_{2,s} \hdots < t_{M,s} \leq T_{s-1}+T < T_s$ and the history stack $\mathcal{H}_{s}$ in \eqref{eq:HS} satisfies
\begin{equation*}
\underline{\sigma}_{\theta} < \lambda_{\min}\left(\Sigma_{\mathcal{Y}_{s}}\right),
\end{equation*}
where $\lambda_{\min}(\cdot)$ denotes the minimum eigenvalue.
\end{assum}
\begin{algorithm}
\small
\caption{\small  Algorithm for adaptive DNN history stack observer with dwell-time. At each time instance $t$, $\eta(t)$ stores the last time $\mathcal{H}$ was purged, $\Lambda(t)$ stores the highest minimum eigenvalue of $\Sigma_{\mathcal{Y}}$ encountered so far}
\label{algo:parameterEstimatorAlgo}

\KwIn{
\begin{tabular}{ll}
$T > 0$ & Final time horizon\\
$\mathcal{T} \in \mathbb{R}_{\geq 0}$ & Minimum dwell-time\\
$\lambda^{*} \in \mathbb{R}_{\geq 0}$ & Eigenvalue threshold\\
$t^{*} > 0$ & Sampling time\\
$\xi \in (0,1]$ & Purging threshold\\
$[\underline{t},\, \overline{t}] \subseteq [0, T]$ & Offline training interval
\end{tabular}
}


\textbf{Initialize:} $t_{0} \gets 0$, $\eta(t_0) \gets 0$, $\Omega_{\max}(t_0) \gets 0$\;

\While{$t_{0} < T$}{
    \uIf{$t < \underline{t}$}{
        Add data pair $(\hat{x}(t), \dot{\hat{x}}(t))$ to training set $\mathscr{D}$\;
    }
    \uElseIf{$\underline{t} \leq t \leq \overline{t}$}{
        Train inner DNN weights using LM algorithm on $\mathscr{D}$\;
        Save trained inner-layer weights $\{W_{j}^{*}\}_{i=1}^{k}$\;
    }
    \ElseIf{$t > \overline{t}$}{
         $W_{j}(t) \gets W_{j}^{*}$ for all $j = 1,\dots,k$\tcp*{Update weights}
    }
    
    \If{$t - \eta(t) > t^{*}$ and new data available}{
        \uIf{$\mathcal{M}$ is not full}{
            Add the data points to $\mathcal{M}$\;
        }
        \Else{
            Add the data points to $\mathcal{M}$ if condition \eqref{eq:eigenValueCondition} holds\;
        }
        
        \If{$\lambda_{\min}(\Sigma_{\mathcal{Y}}) \geq \xi \Lambda(t)$}{
            \If{$t - \eta(t) \geq \mathcal{T}(t)$}{
                $\mathcal{H} \gets \mathcal{M}$ \tcp*{Purge and replace $\mathcal{H}$} 
                $\mathcal{G} \gets 0$ \;
                $\eta(t) \gets t$\;
                $\Lambda(t) \gets \max\left\{\Lambda(t), \lambda_{\min}(\Sigma_{\mathcal{Y}}(t))\right\}$\;
            }
        }
    }
    $t_{0} \gets t$\;
}
\end{algorithm}

Since the rate of convergence of the outer layer DNN weight-estimation errors depends on $\underline{\sigma}_{\theta}$, the minimum eigenvalue maximization algorithm is utilized for the selection of the time instances $\{t_{i,s}\}_{i=1}^{M} \subset \mathcal{I}_{s-1}$ (see, for example, \cite{SCC.Ogri.Bell.ea2023}). 

The DNN is updated using a two-step process. The weights and biases of the input and hidden layers are trained in a parallel thread using the LM algorithm on input-output data pairs $\{\hat{x}(t_{j}), \dot{\hat{x}}(t_{j})\}_{j \in \mathbb{Z}_{\geq 0}}$ sampled along the trajectories
of \eqref{eq:DNNDynamicsEstimate}. The LM algorithm minimizes the mean squared error (MSE) between the input-output data pairs, i.e., $\text{MSE} = \frac{1}{N} \sum_{j=1}^{N} \bigl\| \hat{x}(t_{j+1}) - \dot{\hat{x}}(t_{j}) \bigr\|^{2}$. The output-layer weights are updated in real-time using the ICL-based update law in \eqref{eq:outerLayerWeightsUpdate}. At the start of the training, the outer layer weights from the main thread are copied to the parallel thread and remain fixed during the training of the hidden layers. After training the hidden layers in the parallel thread, the updated inner layer weights are transferred to the main thread and remain fixed until the next training cycle \cite{SCC.Prabhat.Mateus.ea2023}.

In the following section, a Lyapunov-based stability analysis is utilized to inform the choice of $\varepsilon$ in Theorem~\ref{thm:safeControl} and show uniform ultimate boundedness of the trajectories of the closed-loop system under the desired stabilizing control policy in \eqref{eq:uControl}.

\section{Stability Analysis}\label{section:stabilityAnalysis}

Note that the robustified CBF in \eqref{eq:qp} only guarantees safety if the bounds in Assumptions~\ref{ass:boundedGradientF} and \ref{ass:boundedGradientG} are computed over the set $\bar{B}(0, \varepsilon)$ and if the estimation error is contained within $\bar{B}(0, \varepsilon)$. To compute an approximate $\varepsilon \geq 0$, we first show stability of the entire closed-loop system without invoking Theorem~\ref{thm:robustSafety}
and \ref{thm:safeControl}. The stability analysis allows us to bound the estimation error. The resulting bound informs the value of $\varepsilon$ used in Theorem~\ref{thm:safeControl} to ensure safety.

The stability analysis is carried out separately over the time intervals $\mathcal{I}_{s}$ for $s \in \mathbb{N}$. The history stack that is active during $\mathcal{I}_{s}$ is denoted by $\mathcal{H}_{s}$. Note that $\mathcal{H}_{s}$ is recorded over the previous interval $\mathcal{I}_{s-1}$. Over the interval $\mathcal{I}_{s}$, the history stack $\mathcal{H}_{s}$ is not updated, meaning the matrices $\Sigma_{\mathcal{Y}_{s}}$ and $\Sigma_{\mathcal{E}_{s}}$ are constant. The history stack $\mathcal{H}_{1}$ is initialized using arbitrarily selected vectors satisfying $\hat{x} \in \Omega$ such that $\mathcal{H}_{1}$ satisfies Assumption~\ref{ass:regressorRank}\footnote{This arbitrary initialization of $\mathcal{H}_{1}$ may result in a potentially large initial error $\mathcal{E}_{1}$ in \eqref{eq:StateEstimateRelation}, potentially leading to significant weight estimation errors, $\tilde{\theta}$, during $\mathcal{I}_{1}$. However, as long as $\mathcal{H}_{1}$ satisfies Assumption~\ref{ass:regressorRank}, the first term in \eqref{eq:weightErrorDyn} ensures that $\tilde{\theta}$ remains bounded over $\mathcal{I}_{1}$.}. 

Consider the concatenated state $Z \coloneqq \begin{bmatrix} x^{\top} & \tilde{x}^{\top} & \vect(\tilde{\theta})^{\top} & {\tilde{W}_{a}}^{\top} &{\tilde{W}_{c}}^{\top}
\end{bmatrix}^{\top} \in \mathbb{R}^{n(2+p)+2L}$, where $\vect(\cdot)$ is the vectorization operator. The dynamics of $Z$ are governed by \eqref{eq:dynamics}, \eqref{eq:DNNObervererrorDyn}, \eqref{eq:weightErrorDyn}, \eqref{eq:criticUpdate}, and \eqref{eq:actorUpdate}, under the control policy \eqref{eq:qp}.

A summary of the stability analysis, along with a graphical representation in Fig.~\ref{fig:switchingIntervals}, is provided below. The key idea is to show that the system state $Z$ and the state estimation error $\tilde{x}$ become progressively smaller across the intervals $\mathcal{I}_{s}$.

\textit{Interval $\mathcal{I}_{1}$:} The proof begins by showing that the closed-loop state $Z$ is bounded over $\mathcal{I}_{1}$ under the control policy defined in \eqref{eq:qp}, where the bound is $O\left(\left\|Z(0)\right\| + \left\|\sum_{i=1}^M \mathcal{E}_{1i}\right\| + \overline{\epsilon} + \overline{\epsilon}_{\theta}\right)$\footnote{For a positive function $g$, $f = O(g)$ if there exists a constant $M$ such that $\|f(x)\| \leq M\|g(x)\|, \forall x \in \Omega$.}.
 Utilizing this bound, we compute a value of $\varepsilon$ to guarantee safety and establish conditions for selecting gains to ensure that $\|\tilde{x}(T_{1})\| < \overline{\upsilon}^{-1}(\underline{\upsilon}(\overline{\epsilon}_{x}))$ for a given $\overline{\epsilon}_{x} \in \mathbb{R}_{>0}$, where $\underline{\upsilon},\overline{\upsilon}:\mathbb{R}_{\geq 0}\to\mathbb{R}_{\geq 0}$ are class $\mathcal{K}$ functions that satisfy \eqref{eq:VBound}. 
 
 \textit{Interval $\mathcal{I}_{2}$:} The history stack $\mathcal{H}_{2}$, active during the interval $\mathcal{I}_{2}$, is populated with data collected during $\mathcal{I}_{1}$. Without loss of generality, $\mathcal{H}_{2}$ is can be assumed to represent the system better than the arbitrarily selected $\mathcal{H}_{1}$, i.e., $\left\|\sum_{i=1}^M \mathcal{E}_{1i}\right\| \geq \left\|\sum_{i=1}^M \mathcal{E}_{2i}\right\|$. We then show that the bound on $Z$ over $\mathcal{I}_{2}$ is smaller than that over $\mathcal{I}_{1}$, in particular, $\|\tilde{x}(t)\| \leq \overline{\epsilon}_{x}$ for all $t \in \mathcal{I}_{2}$. Furthermore, the same $\varepsilon$ as interval $\mathcal{I}_{1}$ can be used to solve the QP in \eqref{eq:qp} and guarantee safety over the interval $\mathcal{I}_{2}$. 
  
\textit{Interval $\mathcal{I}_{3}$:} Since history stack $\mathcal{H}_{3}$, active over $\mathcal{I}_{3}$, is recorded over $\mathcal{I}_{2}$, and $\|\tilde{x}(t)\| \leq \overline{\epsilon}_{x}$, for all $t \in \mathcal{I}_{2}$, the error $\left\|\sum_{i=1}^M \mathcal{E}_{3i}\right\|$ is shown to be $O(\overline{\epsilon}_{x} + \overline{\epsilon}_{\theta})$. If $T_{3} = \infty$, then it is established that $\limsup_{t \to \infty} \|Z(t)\| = O(\overline{\epsilon}_{x} + \overline{\epsilon} + \overline{\epsilon}_{\theta})$. If $T_{3} < \infty$, then the tighter bound on $Z$ over $\mathcal{I}_{3}$ (compared to $\mathcal{I}_{2}$) is used to show that $\|\tilde{x}(t)\| \leq \overline{\epsilon}_{x}$ for all $t \in \mathcal{I}_{3}$. The analysis then extends inductively to show $\limsup_{t \to \infty} \|Z(t)\| = O(\overline{\epsilon}_{x} + \overline{\epsilon} + \overline{\epsilon}_{\theta})$ and $\|\tilde{x}(t)\| \leq \overline{\epsilon}_{x}$ for all $t \in [T_{2}, \infty)$. The inductive argument also shows that the same $\varepsilon$ as $\mathcal{I}_{1}$ can be selected to solve the QP in \eqref{eq:qp} and ensure safety for all $t \in [T_{2}, \infty)$.

In summary, we construct an inductive argument to show that the norm $\|\mathcal{E}_{s_{i}}\|$ is decreasing across each subsequent interval $\mathcal{I}_{s}$ for $s \in \mathbb{N}$. Leveraging this decrease in $\|\mathcal{E}_{s_{i}}\|$, we derive sufficient gains conditions to ensure safety and convergence of $Z$ to a neighborhood of the origin.

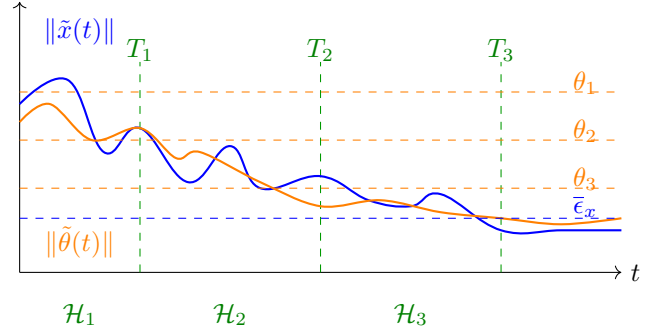
\begin{figure}
  \centering
  \begin{tikzpicture}[scale=0.8]
    \draw[->] (0,0) -- (10,0) node[right] {$t$};
    \draw[->] (0,0) -- (0,4.5) node[above] {};

    \draw[dashed,green!60!black] (2,0) -- (2,3.5);
    \draw[dashed,green!60!black] (5,0) -- (5,3.5);
    \draw[dashed,green!60!black] (8,0) -- (8,3.5);
    \node[green!50!black] at (2,3.7) {$T_{1}$};
    \node[green!50!black] at (5,3.7) {$T_{2}$};
    \node[green!50!black] at (8,3.7) {$T_{3}$};

    \draw[thick,blue] plot[smooth, tension=0.7] coordinates {(0,2.8) (0.8,3.2) (1.4,2.0) (2.0,2.4) (2.8,1.5) (3.5,2.1) (4.0,1.4) (5.0,1.6) (5.8,1.2) (6.5,1.1) (7.0,1.3) (8.0,0.7) (9.0,0.7) (10,0.7)};
    \node[blue] at (1.0,4.0) {$\|\tilde{x}(t)\|$};

    \draw[thick,orange] plot[smooth, tension=0.6] coordinates {(0,2.5) (0.5,2.8) (1.2,2.2) (2.0,2.4) (2.6,1.9) (3.0,2.0) (4.0,1.5) (5.0,1.1) (6.0,1.2) (7.0,1.0) (8.0,0.9) (9.0,0.8) (10,0.9)};
    \node[orange] at (1.0,0.5) {$\|\tilde{\theta}(t)\|$};

    \draw[dashed,blue] (0,0.9) -- (10,0.9);
    \node[blue] at (9.4,1.1) {$\overline{\epsilon}_{x}$};

    \draw[dashed,orange] (0,3) -- (10,3);
    \draw[dashed,orange] (0,2.2) -- (10,2.2);
    \draw[dashed,orange] (0,1.4) -- (10,1.4);
    \node[orange] at (9.4,3.15) {$\theta_{1}$};
    \node[orange] at (9.4,2.35) {$\theta_{2}$};
    \node[orange] at (9.4,1.55) {$\theta_{3}$};

    \node[green!50!black] at (1.0,-0.7) {$\mathcal{H}_{1}$};
    \node[green!50!black] at (3.5,-0.7) {$\mathcal{H}_{2}$};
    \node[green!50!black] at (6.5,-0.7) {$\mathcal{H}_{3}$};
  \end{tikzpicture}
  \caption{Trajectory of error signals over switching intervals.}
  \label{fig:switchingIntervals}
\end{figure}

The following theorem establishes local uniform ultimate boundedness (LUUB) of the closed-loop system.
\begin{thm}\label{thm:boundedAndSafe}
 Consider an arbitrary positive constant $\overline{\epsilon}_{x} > 0$, the dynamics \eqref{eq:dynamics}, the state observer \eqref{eq:DNNDynamicsEstimate}, the adaptive update law \eqref{eq:outerLayerWeightsUpdate}, and the control policy \eqref{eq:uControl}. Provided Assumptions~\ref{ass:boundedfunctions} through \ref{ass:boundedGradientG} hold, if the weights $\hat{\theta}$, $\hat{W}_{c}$, $\Gamma$, and $\hat{W}_{a}$ are updated according to \eqref{eq:outerLayerWeightsUpdate}, \eqref{eq:criticUpdate}, \eqref{eq:gammaUpdate}, and \eqref{eq:actorUpdate}, respectively, the Bellman extrapolation trajectories are selected such that they satisfy Assumption~\ref{ass:rankCond} and $x_{k}(t) \in \Omega$ for all $t \in \mathbb{R}_{\geq 0}$ and $k = 1, \hdots, N$, the history stacks $\mathcal{H}$ and $\mathcal{M}$ are maintained and updated using Algorithm~\ref{algo:parameterEstimatorAlgo}, the control gains are selected to satisfy the sufficient conditions in \eqref{eq:cond1}--\eqref{eq:cond3}, \eqref{eq:cond4}, and \eqref{eq:cond5}, the dwell time $\mathcal{T}$, is selected such that $\mathcal{T}(t) = \mathcal{T}_{s}$, where $\mathcal{T}_{s}$ is selected large enough to satisfy \eqref{eq:finalDwellTime}, and the excitation interval is sufficiently long such that $T_{2} < T$, then the concatenated state $Z(t)$ is LUUB. Furthermore, if the QP in \eqref{eq:qp} is solved by selecting $\varepsilon$ according to \eqref{eq:epsilonChoice}, then the system in \eqref{eq:closedSystem} is guaranteed to be safe with respect to the sets $(\bar{B}(\hat{x}_{0}, \varepsilon), \mathcal{S})$ for all $t \in \mathbb{R}_{\geq 0}$.
\end{thm}
\begin{proof}
Consider a continuously differentiable candidate Lyapunov function $\mathcal{V}: \mathbb{R}^{n(2+p)+2L} \times \mathbb{R}_{\geq 0} \to \mathbb{R}$, defined as
\begin{multline}\label{eq:Lyap}
\mathcal{V}\left(Z,t\right) \coloneqq V^{*}(x) + \tilde{x}^{\top}P \tilde{x} + \frac{1}{2}\tr\left(\tilde{\theta}^{\top}\gamma^{-1}\tilde{\theta}\right) \\+ \frac{1}{2}\tilde{W}_{c}^{\top} \Gamma^{-1}(t)\tilde{W}_{c} + \frac{1}{2}\tilde{W}_{a}^{\top}\tilde{W}_{a},
\end{multline}
where $P$ is the solution to the Lyapunov equation below \eqref{eq:DNNDynamicsEstimate} and $\tr(\cdot)$ represents the trace operator. 

We begin by developing bounds that hold regardless of the history stack being employed for parameter estimation. As described in \cite[Lemma~1]{SCC.Kamalapurkar.Rosenfeld.ea2016}, provided Assumption~\ref{ass:rankCond} holds and $\lambda_{\min}\{\Gamma^{-1}(0)\}> 0$, the update laws in \eqref{eq:criticUpdate} and \eqref{eq:gammaUpdate} ensure that the least squares update law satisfies
\begin{equation}\label{eq:Gammabound}
\underline{\Gamma}I_{L}\leq\Gamma(t)\leq\overline{\Gamma}I_{L},
\end{equation}
$\forall t \in \mathbb{R}_{\geq 0}$, where $\overline{\Gamma},\underline{\Gamma} \in \mathbb{R}_{> 0}$ are constants, and $I_{L}$ is an $L$ by $L$ identity matrix. Since the candidate Lyapunov function $\mathcal{V}$ is positive definite and decrescent, the bound in \eqref{eq:Gammabound} can be used to conclude that it is bounded as \cite[Lemma 4.3]{SCC.Khalil2002}
\begin{equation}\label{eq:VBound}
\underline{\upsilon}\left(\left\Vert Z\right\Vert \right)\leq \mathcal{V}\left(Z,t\right)\leq\overline{\upsilon}\left(\left\Vert Z\right\Vert \right),
\end{equation}
for all $t \in \mathbb{R}_{\geq 0}$ and for all $Z\in\mathbb{R}^{n(2+p)+2L}$, where $\underline{\upsilon},\overline{\upsilon}:\mathbb{R}_{\geq 0}\to\mathbb{R}_{\geq 0}$ are class $\mathcal{K}$ functions. Using \eqref{eq:gammaUpdate} and \eqref{eq:Gammabound}, the normalized regressor $\frac{\omega_{k}}{\rho_{k}}$ is bounded as $\sup_{t \in \mathbb{R}_{\geq 0}}\|\frac{\omega_{k}}{\rho_{k}}\| \leq \frac{1}{2\sqrt{\nu\underline{\Gamma}}}$. 
 The Lie derivative of the optimal value function $V^{*}$, along the flow of \eqref{eq:dynamics}, \eqref{eq:criticUpdate}, \eqref{eq:gammaUpdate}, and \eqref{eq:actorUpdate} under the control law in \eqref{eq:qp} is given by
 \begin{equation}\label{eq:valueDiff}
     \dot{V}^{*}\left(Z, t\right) =  {\nabla}_{x}{V^{*}}(x)\Big(f(x)+ g(x)\pi(\hat{x}, t)\Big).
 \end{equation}
Substituting the safe controller from \eqref{eq:qp} into \eqref{eq:valueDiff}, and then adding and subtracting the optimal control policy \eqref{eq:optimalcontrol} and the approximate control policy \eqref{eq:uControl}, yields
\begin{multline}
     \dot{V}^{*}\left(Z, t\right) =  {\nabla}_{x}{V^{*}}(x)\Big(f(x) + g(x)u^{*}(x)\Big)\\
     + {\nabla}_{x}{V^{*}}(x)g(x)\left(\pi_{\operatorname{des}}(\hat{x}, t) - u^{*}(x)\right) \\ + {\nabla}_{x}{V^{*}}(x)g(x)\left(\pi(\hat{x}, t) -\pi_{\operatorname{des}}(\hat{x}, t) \right).
\end{multline}
Let $\mathcal{D} \subset \mathbb{R}^{n(2+p)+L}$ be a compact set defined as $\mathcal{D} \coloneqq \Omega \times (\Omega \oplus \Omega) \times \bar{B}(0, \overline{\theta}) \times \bar{B}(0, \overline{W})$,
where the notation $\oplus$ denotes the Minkowski sum.
By substituting \eqref{eq:HJB}, \eqref{eq:optimalcontrol}, \eqref{eq:uApprox}, \eqref{eq:gammaUpdate}, and using the bound in \eqref{eq:controlBound}, the Lie derivative of the value function $V^{*}$ is bounded on the set $\mathcal{D} \times \mathbb{R}^{L} \times \mathbb{R}_{\geq 0}$ as
\begin{equation}\label{eq:ValueDerivBound}
     \dot{V}^{*}\left(Z, t\right) \leq  -\underline{q}\left(\|x\|\right)
       + \ell_{1}\left\|\tilde{x}\right\| \\
       + \ell_{2}\left\|\tilde{W}_{a}\right\| + \vartheta_{0}, 
\end{equation}
where $\ell_{1} \coloneqq \frac{(\overline{W}^{2}\overline{\nabla\sigma} + \overline{W}\overline{\epsilon})\overline{g}L_{g\sigma}}{2\lambda_{\max}(R)}$, $\ell_{2} \coloneqq \frac{(\overline{W}\overline{\nabla\sigma} + \overline{\epsilon})\overline{g}^{2}\overline{\nabla\sigma}}{2\lambda_{\max}(R)}$, $\vartheta_{0} \coloneqq \overline{g}\epsilon_{\pi}(\overline{W}\nabla_{\sigma} + \overline{\epsilon})$, and $L_{g\sigma}$ denotes the Lipschitz constant for the mapping $x \mapsto g(x)^{\top}\nabla_{x}\sigma(x)^{\top}$ on $\Omega$.

Let $\mathcal{W}(\tilde{W}_{c}, \tilde{W}_{a}, t) \coloneqq \frac{1}{2}\tilde{W}_{c}^{\top} \Gamma^{-1}(t)\tilde{W}_{c} + \frac{1}{2}\tilde{W}_{a}^{\top}\tilde{W}_{a}$. The lie derivative of $\mathcal{W}$ along the flow of \eqref{eq:criticUpdate}--\eqref{eq:actorUpdate} is given by
{\medmuskip=0mu\thinmuskip=0mu\thickmuskip=0mu\begin{equation}\label{eq:weightsLyapDeriv}
     \dot{\mathcal{W}}(Z, t) = - \tilde{W}_{c}^{\top} \Gamma^{-1}\dot{\hat{W}}_{c} -\frac{1}{2} \tilde{W}_{c}^{\top} \Gamma^{-1}\dot{\Gamma}\Gamma^{-1}\tilde{W}_{c} -\tilde{W}_{a}^{\top}\dot{\hat{W}}_{a},
 \end{equation}}with $\dot{\tilde{W}}_{c} = -\dot{\hat{W}}_{c}$, and $\dot{\tilde{W}}_{a} =  -\dot{\hat{W}}_{a}$. By the property of projection operators in \cite[Lemma E.1. IV]{SCC.Krstic.Kanellakopoulos.ea1995}, it holds that $-\tilde{W}_{a}^{\top}\proj_{\bar{B}(0,\overline{W})}\Bigl(f_{a}(\hat{W}_{a}, \hat{W}_{c})\Bigr) \leq - \tilde{W}_{a}^{\top} f_{a}(\hat{W}_{a}, \hat{W}_{c})$. Substituting the update laws in \eqref{eq:criticUpdate}, \eqref{eq:gammaUpdate}, and \eqref{eq:actorUpdate} into \eqref{eq:weightsLyapDeriv},
 since Assumption~\ref{ass:rankCond} holds and the Bellman extrapolation trajectories are selected such that they satisfy $x_{k}(t) \in \Omega$ for all $t \in \mathbb{R}_{\geq 0}$ and $k = 1, \hdots, N$, the Lie derivative of $\mathcal{W}$ is bounded on the set $\mathcal{D} \times \mathbb{R}^{L} \times \mathbb{R}_{\geq 0}$ as
\begin{multline}\label{eq:weightLyapDerivBound} \dot{\mathcal{W}}\left(Z, t\right) \leq -k_{c}\underline{c}\left\|\tilde{W}_{c}\right\|^{2}  - \left(k_{a_{1}}+k_{a_{2}}\right)\left\|\tilde{W}_{a}\right\|^{2}  + \ell_{3}\left\|\tilde{W}_{a}\right\|^{2}
      \\+ \ell_{4}\left\|\tilde{W}_{c}\right\|  + \ell_{5}\left\|\tilde{W}_{a}\right\|\left\|\tilde{W}_{c}\right\| + \ell_{6}\left\|\tilde{W}_{a}\right\|,
\end{multline}
 where $\underline{c} \coloneqq\frac{\beta}{2\overline{\Gamma}k_{c}} + \frac{\underline{c}_{1}}{2}$, $\ell_{3} \coloneqq \frac{k_{c}\overline{G_{\sigma}}_{k}\overline{W}}{8\sqrt{\nu\underline{\Gamma}}}$, $\ell_{4} \coloneqq \frac{k_{c}\overline{\Delta}_{k}}{2\sqrt{\nu\underline{\Gamma}}}$, $\ell_{5} \coloneqq \frac{k_{c}\overline{G_{\sigma}}_{k}\overline{W}}{8\sqrt{\nu\underline{\Gamma}}}+k_{a_{1}}$, and $\ell_{6} \coloneqq \frac{k_{c}\overline{G_{\sigma}}_{k} \overline{W}}{8\sqrt{\nu\underline{\Gamma}}}^{2}+ k_{a_{2}}\overline{W}$.
 
 To facilitate proof by mathematical induction, we now develop bounds that hold over a generic interval $\mathcal{I}_{s}$ assuming that the history stack $\mathcal{H}_{s}$, recorded over $\mathcal{I}_{s-1}$, satisfies Assumption~\ref{ass:regressorRank} and that the states $(x, \tilde{x}, \tilde{\theta}, \tilde{W}_{a})$ remain within the compact set $\mathcal{D}$ over the time interval $\mathcal{I}_{s-1}$. Let $\Theta(Z, t) \coloneqq  \tilde{x}^{\top}P \tilde{x} + \frac{1}{2}\tr(\tilde{\theta}^{\top}\gamma^{-1}\tilde{\theta})$. By the property of projection operators in \cite[Lemma E.1. IV]{SCC.Krstic.Kanellakopoulos.ea1995}, it holds that $-\tilde{\theta}^{\top} \gamma^{-1} \proj_{\bar{B}(0, \overline{\theta})}(k_{\theta}\gamma\Sigma_{\mathcal{Y}_{s}}\tilde{\theta} + k_{\theta}\gamma\Sigma_{\mathcal{E}_{s}}) \leq - \tilde{\theta}^{\top} \gamma^{-1}(k_{\theta}\gamma\Sigma_{\mathcal{Y}_{s}}\tilde{\theta} + k_{\theta}\gamma\Sigma_{\mathcal{E}_{s}})$. Since $(A - KC)^{\top}P + P( A - KC) = - S$, the Lie derivative of $\Theta$ along the flow of \eqref{eq:DNNObervererrorDyn} and \eqref{eq:weightErrorDyn} is bounded over the time interval $\mathcal{I}_{s}$ as
{\medmuskip=1mu\thinmuskip=1mu\thickmuskip=1mu\begin{multline}\label{eq:LyapD}
    \dot{\Theta}_{s}(Z, t) \leq -\tilde{x}^{\top}S\tilde{x} +  2\tilde{x}^{\top}P\tilde{\theta}^{\top}\phi\left(\Phi(\hat{x}, t)\right) + 2\tilde{x}^{\top}Pe(x, \hat{x}, u, t) \\ -k_{\theta}\tr\left(\tilde{\theta}^{\top}\Sigma_{\mathcal{Y}_{s}}\tilde{\theta}\right) - k_{\theta}\tr\left(\tilde{\theta}^{\top}\Sigma_{\mathcal{E}_{s}}\right).
\end{multline}}
 Using the bound in \eqref{eq:controlBound}, the error $(x, \hat{x}, u) \mapsto e(x, \hat{x}, u, t)$ is bounded on the set $\mathcal{D}$ as
\begin{equation}\label{eq:errorBound}
    e(x, \hat{x}, u, t) \leq \vartheta_{1}\left\|\tilde{x}\right\| + \vartheta_{2}\left\|\tilde{W}_{a}\right\| + \overline{\epsilon_{\theta}},
\end{equation}
where $\vartheta_{1} \coloneqq \overline{\theta}\overline{\nabla\phi}L_{\Phi} + \frac{1}{2}L_{gR\sigma}\overline{W}  + \epsilon_{\pi}L_{g}$, $\vartheta_{2} \coloneqq \frac{\overline{g}^{2}\overline{\nabla\sigma}}{\lambda_{\max}(R)}$, and $L_{gR\sigma}$ is a Lipschitz constant for the mapping $x \mapsto g(x)R^{-1}g(x)^{\top}\nabla_{x}\sigma(x)^{\top}$ on $\Omega$. 

Since the states $(x, \tilde{x}, \tilde{\theta}, \tilde{W}_{a})$ remain within the compact set $\mathcal{D}$ over the time interval $\mathcal{I}_{s-1}$, the residual error terms $\mathcal{E}(t_{i, s})$ in Lemma~\ref{lem:ErrorTermformulation} are also bounded as
\begingroup\medmuskip=0mu\thinmuskip=2mu\thickmuskip=2mu\begin{equation}\label{eq:epsilonBound}
    \|\mathcal{E}(t_{i, s})\| \leq \ell_{x}\overline{\tilde{x}}_{\mathcal{I}_{s-1}} +  \ell_{w}\overline{\tilde{W}}_{a_{\mathcal{I}_{s-1}}} + \ell_{\epsilon} \overline{\epsilon}_{\theta}, \quad \forall j= 1, \hdots M,
\end{equation}\endgroup
for some positive constants $\ell_{x}$, $\ell_{w}$, and $\ell_{\epsilon} \in \mathbb{R}_{> 0}$, where $\overline{\tilde{x}}_{\mathcal{I}_{s-1}} \coloneqq \max_{i \in \{1, \hdots, M\}}\sup_{t \in \mathcal{I}_{s-1}}\left\|\tilde{x}(t)\right\|$ and $\overline{\tilde{W}}_{a_{\mathcal{I}_{s-1}}} \coloneqq \max_{i \in \{1, \hdots, M\}}\sup_{t \in \mathcal{I}_{s-1}}\left\|\tilde{W}_{a}(t)\right\|$. Thus, there exist some positive constant $\overline{\mathcal{E}}
_{s} > 0$ such that $\|\mathcal{E}_{s_{i}}\| \leq \overline{\mathcal{E}}
_{s}$ on the interval $\mathcal{I}_{s-1}$ for $s \in \mathbb{N}$.

Since $\mathcal{H}_{s}$ is recorded over the interval $\mathcal{I}_{s-1}$ and $\left\|\frac{\mathcal{Y}_{s_{i}}}{1+\kappa\|\mathcal{Y}_{s_{i}}\|^{2}}\right\| \leq \frac{1}{2\sqrt{\kappa}} ,\forall i \in \{1,\hdots,M\}$, if Assumption~\ref{ass:regressorRank} is satisfied, then the Lie derivative of $\Theta$ is bounded on the set $\mathcal{D} \times \mathbb{R}^{L} \times \mathcal{I}_{s}$ as  
\begin{multline}\label{eq:StateParameterErrorLyapDerivBound}
    \dot{\Theta}_{s}(Z, t) \leq -\lambda_{\min}\left(S\right)\left\|\tilde{x}\right\|^{2} -k_{\theta}\underline{\sigma}_{\theta}\left\|\tilde{\theta}\right\|^{2} + \ell_{7}\left\|\tilde{x}\right\|^{2} + \ell_{8}\left\|\tilde{x}\right\| \\+ \ell_{9}\overline{\mathcal{E}}
_{s}\left\|\tilde{\theta}\right\| + \ell_{10}\left\|\tilde{x}\right\|\left\|\tilde{\theta}\right\| + \ell_{11}\left\|\tilde{x}\right\|\left\|\tilde{W}_{a}\right\|,
\end{multline}
where $\ell_{7} \coloneqq 2\lambda_{\max}\left(P\right)\vartheta_{1}$, $\ell_{8} \coloneqq 2\lambda_{\max}\left(P\right)\overline{\epsilon_{\theta}}$, $\ell_{9} \coloneqq \frac{k_{\theta}M}{2\sqrt{\kappa}}$, $\ell_{10} \coloneqq 2\lambda_{\max}(P)\overline{\phi}$, and $\ell_{11} \coloneqq 2\lambda_{\max}(P)\vartheta_{2}$.

Using the bounds on the Lie derivatives in \eqref{eq:ValueDerivBound}, \eqref{eq:weightLyapDerivBound}, \eqref{eq:StateParameterErrorLyapDerivBound},
and provided the gain conditions \begin{align}
    & \lambda_{\min}\left(S\right) \geq 5\ell_{7},\label{eq:cond1}\\
    &k_{\theta} \geq \frac{15\ell_{10}^{2}}{4\underline{\sigma}_{\theta}\lambda_{\min}\left(S\right)}, \text{ and}\label{eq:cond2}\\
    & \left(k_{a_{1}}+k_{a_{2}}\right) \geq \frac{9\ell_{5}^{2}}{4k_{c}\underline{c}} + \frac{15\ell_{11}^{2}}{4\lambda_{\min}\left(S\right)} + 3\ell_{3},  \label{eq:cond3}
\end{align} are satisfied, then the Lie derivative of the candidate Lyapunov function, $\mathcal{V}$, along the flow of \eqref{eq:dynamics}, \eqref{eq:DNNObervererrorDyn}, \eqref{eq:weightErrorDyn}, \eqref{eq:criticUpdate}, \eqref{eq:gammaUpdate}, and \eqref{eq:actorUpdate} is bounded on the set $\mathcal{D} \times \mathbb{R}^{L} \times \mathcal{I}_{s}$, as
\begin{multline}\label{eq:VorbitDeriv}
     \dot{\mathcal{V}}_{s}(Z,t) \leq-\underline{q}\left(\|x\|\right)-\frac{\lambda_{\min}\left(S\right)}{5}\left\|\tilde{x}\right\|^{2} - \frac{k_{\theta}\underline{\sigma}_{\theta}}{3}\left\|\tilde{\theta}\right\|^{2} \\ -\frac{k_{c}\underline{c}}{3}\left\|\tilde{W}_{c}\right\|^{2}-\frac{\left(k_{a_{1}}+k_{a_{2}}\right)}{3}\left\|\tilde{W}_{a}\right\|^{2} + \iota_{s},
 \end{multline}
 where $\iota_{s} \coloneqq \frac{5(\ell_{1}+\ell_{8})^{2}}{4\lambda_{\min}\left(S\right)} + \frac{3\ell_{9}^{2}\overline{\mathcal{E}}
_{s}^{2}}{4k_{\theta}\underline{\sigma}_{\theta}} + \frac{3\ell_{4}^{2}}{4k_{c}\underline{c}} + \frac{3\left(\ell_{2} + \ell_{6}\right)^{2}}{4\left(k_{a_{1}}+k_{a_{2}}\right)} + \vartheta_{0}$. Let $\mu: \mathbb{R}_{\geq 0} \to \mathbb{R}_{\geq 0}$ be a class $\mathcal{K}_{\infty}$ function that satisfies
\begin{multline}\label{eq:upsilonClassK}
     \mu\left(\left\|Z\right\|\right) \leq \frac{\underline{q}\left(\|x\|\right)}{2}+\frac{\lambda_{\min}\left(S\right)}{10}\left\|\tilde{x}\right\|^{2} + \frac{k_{\theta}\underline{\sigma}_{\theta}}{6}\left\|\tilde{\theta}\right\|^{2} \\ +\frac{k_{c}\underline{c}}{6}\left\|\tilde{W}_{c}\right\|^{2} + \frac{\left(k_{a_{1}}+k_{a_{2}}\right)}{6}\left\|\tilde{W}_{a}\right\|^{2}.
\end{multline}
Such a $\mu \in \mathcal{K}_{\infty}$ exists since $\underline{q}$ is class $\mathcal{K}_{\infty}$ by assumption. Using \eqref{eq:upsilonClassK}, the Lie derivative of the candidate Lyapunov function $\mathcal{V}$ satisfies $\dot{\mathcal{V}}_{s}(Z,t) \leq -\mu\left(\left\|Z\right\|\right)$, for all $Z \in \mathcal{D} \times \mathbb{R}^{L}$ and for all $t \in \mathcal{I}_{s}$. Let $\chi > 0$ be such that $\bar{B}(0, \chi) \subset \mathcal{D} \times \mathbb{R}^{L}$. If the sufficient condition
\begin{equation}\label{eq:cond4}
  \mu^{-1}(\iota_{s}) \leq  \overline{\upsilon}^{-1}\left(\underline{\upsilon}(\chi)\right)\footnote{Satisfaction of \eqref{eq:cond4} is possible by appropriately selecting points for BE extrapolation such that $\underline{c}$ and control gains ($k_{c}, k_{a_{1}}$, and $k_{a_{2}}$) are sufficiently large enough, and choosing the basis for the value function approximation make $\overline{\epsilon}$ sufficiently small.},
\end{equation} is met, then \cite[Theorem~4.18]{SCC.Khalil2002} can be invoked to conclude that there exist a class $\mathcal{KL}$ function $\beta_{s}$ such that for every initial condition $Z(T_{s-1})$, satisfying $\|Z(T_{s-1})\| \leq \overline{\upsilon}^{-1}(\underline{\upsilon}(\chi))$, the state of the closed loop system is bounded on $\mathcal{I}_{s}$ as
\begin{equation}\label{eq:stateBound}
\|Z(t)\| \leq \max\Bigl\{\beta_{s}(\|Z(T_{s-1})\|, t - T_{s-1}),\Upsilon(\iota_{s})\Bigr\},
\end{equation}
where $\Upsilon(\iota_{s}) = \underline{\upsilon}^{-1}(\overline{\upsilon}(\mu^{-1}(\iota_{s})))$ is the ultimate bound. 

Let $T_{s, \operatorname{settle}}$ represent the settling time for the state $Z$ to reach the ultimate bound $\Upsilon(\iota_{s})$, i.e., $\beta_{s}(\|Z(T_{s-1})\|, T_{s, \operatorname{settle}}) \leq \Upsilon(\iota_{s})$. Since $\iota_{s}$ depends on $\mathcal{E}_{s}$, it can be made smaller by minimizing the estimation errors $\tilde{x}(t)$ associated with the data stored in the history stack. 

Now we begin the mathematical induction argument. Since we populate the history stack $\mathcal{H}_{1}$ using arbitrary vectors $\hat{x}$ that satisfy $\hat{x} \in \Omega$, if the gains are selected such that $\iota_{1}$ satisfies \eqref{eq:cond4}, i.e., $\mu^{-1}(\iota_{1}) \leq  \overline{\upsilon}^{-1}\left(\underline{\upsilon}(\chi)\right)$, then \eqref{eq:stateBound} holds on $\mathcal{I}_{1}$. Let $\overline{\beta_{1}} \coloneqq \sup_{t \in \mathcal{I}_{1}}\{\beta_{1}(\|Z(0)\|, t)\}$ and $\theta_{1} \coloneqq \max\left\{\overline{\beta_{1}}, \Upsilon(\iota_{1})\right\}$. Since $\|\tilde{\theta}(t)\|^{2} \leq \|Z(t)\|^{2}$, it follows from \eqref{eq:stateBound} that  
\begin{equation}\label{eq:theta1}
    \|\tilde{\theta}(t)\| \leq  \theta_{1}, \quad \forall t \in \mathcal{I}_{1}.
\end{equation}

The following analysis establishes a bound on the state estimation error that is sharper than the bound in \eqref{eq:stateBound}. Consider the positive definite candidate Lyapunov function
\begin{equation}\label{eq:Lyap2}
\Psi(\tilde{x}) \coloneqq \tilde{x}^{\top}P \tilde{x},
\end{equation}
which satisfies the inequality $\lambda_{\min}(P)\|\tilde{x}\|^{2}\leq \Psi\left(\tilde{x}\right)\leq \lambda_{\max}(P)\|\tilde{x}\|^{2}$. The Lie derivative of $\Psi$ along the flow of \eqref{eq:DNNObervererrorDyn} is given by 
\begin{multline}\label{eq:lyapError}
    \dot{\Psi}(\tilde{x}, t) = -\tilde{x}^{\top}S\tilde{x} +  2\tilde{x}^{\top}P\tilde{\theta}^{\top}(t)\phi\left(\Phi(\hat{x}(t), t)\right) \\ + 2\tilde{x}^{\top}Pe(x(t), \hat{x}(t), u(t), t).
\end{multline} 
If the outer layer weight estimation error $\tilde{\theta}(t)$ is bounded over $\mathcal{I}_{1}$, then the Lie derivative of $\Psi$ is bounded over $\mathcal{I}_{1}$ as
\begin{multline}\label{eq:WorbitDeriv1}
     \dot{\Psi}_{1}(\tilde{x}, t) \leq -\lambda_{\min}\left(S\right)\left\|\tilde{x}\right\|^{2}  + \ell_{7}\left\|\tilde{x}\right\|^{2} \\+ (\ell_{8}+ \ell_{10}\theta_{1} + \ell_{11}\overline{W})\left\|\tilde{x}\right\|.
 \end{multline}

 Provided the gain condition
\begin{equation}
 \lambda_{\min}\left(S\right) > 3\ell_{7} \label{eq:cond5}
\end{equation} is satisfied, the Lie derivative of $\Psi$ is bounded on the interval $\mathcal{I}_{1}$ as $\dot{\Psi}_{1}(\tilde{x},t)\leq -\varkappa\Psi(\tilde{x}) + \varsigma_{1}$, where $\varkappa \coloneqq \frac{\lambda_{\min}(S)}{6\lambda_{\max}(P)}$ and $\varsigma_{1} \coloneqq \frac{3(\ell_{8}+ \ell_{10}\theta_{1} + \ell_{11}\overline{W})^{2}}{4\lambda_{\min}\left(S\right)}$. Then by the Comparison Lemma \cite[Lemma~3.4]{SCC.Khalil2002}, we can conclude that
\begin{equation}\label{eq:WBound2}
    \Psi(\tilde{x}(t), t) \leq \left(\overline{\Psi}_{1} - \frac{\varsigma_{1}}{\varkappa}\right)\exp(-\varkappa t) +\frac{\varsigma_{1}}{\varkappa},
\end{equation}
for all $t \in \mathcal{I}_{1}$, where $\overline{\Psi}_{1} > 0$ is a constant such that $\vert\Psi(\tilde{x}(0), 0)\vert \leq \overline{\Psi}_{1}$. In particular, it can be concluded that for all $t \in \mathcal{I}_{1}$
\begin{equation}\label{eq:XBound1}
     \|\tilde{x}(t)\| \leq \sqrt{\frac{\lambda_{\max}(P)}{\lambda_{\min}(P)}\max\left\{\overline{\Psi}_{1}, \frac{\varsigma_{1}}{\varkappa}\right\}} \eqqcolon \overline{\tilde{x}}_{1}.
\end{equation}
Selecting the dwell time $\mathcal{T}_{1}$ large enough to satisfy $\mathcal{T}_{1} \geq \max\left\{ -\frac{1}{\varkappa} \ln\left( \frac{\varsigma_{1}/\varkappa}{\overline{\Psi}_{1} - \varsigma_{1}/\varkappa} \right),\; T_{1,\mathrm{settle}} \right\}$, it can be concluded that $\vert\Psi(\tilde{x}(T_{1}), T_{1})\vert \leq \frac{2\varsigma_{1}}{\varkappa}$. As a result, $\|\tilde{x}(T_{1})\| \leq \sqrt{\frac{\lambda_{\max}(P)}{\lambda_{\min}(P)}\left(\frac{2\varsigma_{1}}{\varkappa}\right)}$. Furthermore, provided the dwell time condition for $\mathcal{T}_{1}$ is satisfied, $\|Z(T_{1})\| \leq \Upsilon(\iota_{1})$. Note that the bound on $\tilde{x}(T_{1})$ can be made arbitrarily small by choosing the positive definite matrix $S$ in \eqref{eq:cond5} with a sufficiently large minimum eigenvalue.

We now discuss safety over the interval $\mathcal{I}_{1}$. While the bound in \eqref{eq:stateBound} holds, it does not guarantee safety since $\mathcal{S} \subseteq \Omega$. For the first interval $\mathcal{I}_{1}$, we select 
 \begin{equation}\label{eq:epsilonChoice}
     \varepsilon \coloneqq \sqrt{\frac{\lambda_{\max}(P)}{\lambda_{\min}(P)}\left(\max\left\{\overline{\Psi}_{1}, \frac{2\varsigma_{1}}{\varkappa}\right\}\right)}
 \end{equation}
 for the QP in \eqref{eq:qp}. To ensure safety over $\mathcal{I}_{1}$, Theorem~\ref{thm:safeControl} requires $\bar{B}(\hat{x}_{0}, \varepsilon) \subset \mathcal{S}$. For a given $\hat{x}_0$, let $\bar\varepsilon\coloneqq\sup(\varepsilon\mid\bar B(\hat x_0,\varepsilon)\subset \mathcal{S}$. Safety over $\mathcal{I}_1$ can then be guaranteed if $\varepsilon \leq \bar\varepsilon$, or equivalently, if $\overline{\Psi}_1 \leq \frac{\lambda_{\min}(P)}{\lambda_{\max}(P)}\bar\varepsilon^2$ and $\frac{2\varsigma_{1}}{\varkappa} \leq \frac{\lambda_{\min}(P)}{\lambda_{\max}(P)}\bar\varepsilon^2$. The former provides a set of admissible initial guesses for state estimation and the latter can be ensured by selecting $S$ in \eqref{eq:cond5} with a sufficiently large minimum eigenvalue. Consequently, the conditions of Theorem~\ref{thm:safeControl} are satisfied and the controller guarantees safety with respect to $(\bar{B}(\hat{x}_{0}, \varepsilon), \mathcal{S})$ for all $t \in \mathcal{I}_{1}$.

Next, we extend the analysis to the interval $\mathcal{I}_{2}$. Given that the history stack $\mathcal{H}_{2}$ active during $\mathcal{I}_{2}$ is recorded during $\mathcal{I}_{1}$, the bound in \eqref{eq:epsilonBound} can be used to show that the data used to construct the history stack $\mathcal{H}_{2}$ satisfy $\overline{\mathcal{E}}_{2} = O(\overline{\tilde{x}}_{1} + \epsilon_{\theta})$. 

Given that the history stack $\mathcal{H}_{2}$ provides a more accurate representation of the system dynamics than the arbitrarily selected $\mathcal{H}_{1}$, which is independent of the system trajectories, $\overline{\mathcal{E}}_{1}$ can be selected, without loss of generality, such that $\overline{\mathcal{E}}_{2} < \overline{\mathcal{E}}_{1}$, and as a result, $\iota_{2} < \iota_{1}$. Thus, given a sufficiently large constant $\overline{\mathcal{V}}_{1}$ satisfying $\|\mathcal{V}(Z(0), 0)\| \leq \overline{\mathcal{V}}_{1}$ and $2\underline{\upsilon}(\Upsilon(\iota_{1})) < \overline{\mathcal{V}}_{1}$, it follows that
$\mathcal{V}(Z(t), t) \leq \underline{\upsilon}(\Upsilon(\iota_{1})) \eqqcolon \overline{\mathcal{V}}_{2} < \overline{\mathcal{V}}_{1}, \, \forall t \in \mathcal{I}_{2}$.
The gain conditions in \eqref{eq:cond1}-\eqref{eq:cond3} thus hold over $\mathcal{I}_{2}$. Consequently, we can conclude that the bound
\begin{equation}\label{eq:theta2}
    \|Z(t)\| \leq \underline{\upsilon}^{-1}(\overline{\mathcal{V}}_{2})  \Rightarrow  \|\theta(t)\| \leq \underline{\upsilon}^{-1}(\overline{\mathcal{V}}_{2}) \eqqcolon \theta_{2}
\end{equation}
holds for all $t \in \mathcal{I}_{2}$.
Since $\theta_{2} < \theta_{1}$, all terms involving $\|\tilde{\theta}\|$ (e.g., in \eqref{eq:WorbitDeriv1}) are strictly smaller, and as such, if the gain condition \eqref{eq:cond5} remains satisfied, we can conclude using a Lyapunov analysis similar to \eqref{eq:Lyap2}--\eqref{eq:WBound2} that $\forall t \in \mathcal{I}_{2}$,
\begin{equation}\label{eq:XBound2}
    \|\tilde{x}(t)\| \leq \sqrt{\frac{\lambda_{\max}(P)}{\lambda_{\min}(P)}\left(\max\left\{\overline{\Psi}_{2}, \frac{\varsigma_{2}}{\varkappa}\right\}\right)} \eqqcolon \overline{\tilde{x}}_{2}, 
\end{equation}
where $\varsigma_{2} \coloneqq \frac{3(\ell_{8}+ \ell_{10}\theta_{2} + \ell_{11}\overline{W})^{2}}{4\lambda_{\min}\left(S\right)}$ and $\overline{\Psi}_{2} = \frac{2\varsigma_{1}}{\varkappa}$. 
Since $\varsigma_{2} < \varsigma_{1}$, safety over the interval $\mathcal{I}_{2}$ follows provided $\varepsilon$ from \eqref{eq:epsilonChoice} is used to solve the QP in \eqref{eq:qp}.

Now, consider the interval $\mathcal{I}_{3}$. Since $\mathcal{H}_{3}$ is active during $\mathcal{I}_{3}$ and recorded during $\mathcal{I}_{2}$, the bounds in \eqref{eq:epsilonBound} and \eqref{eq:XBound2} can be used to show that $\overline{\mathcal{E}}_{3} \coloneqq O(\overline{\tilde{x}}_{2} + \overline{\epsilon}_{\theta})$. By selecting the positive definite matrix $S$ in \eqref{eq:cond5} with a sufficiently large minimum eigenvalue, we ensure $\varsigma_{1}$ is small enough, thus making $\overline{\Psi}_{2}$ sufficiently small such that \eqref{eq:XBound2} satisfies $\overline{\tilde{x}}_{2} < \overline{\tilde{x}}_{1}$. As a result, $\overline{\mathcal{E}}_{3} < \overline{\mathcal{E}}_{2}$ which implies $\iota_{3} < \iota_{2}$. If the dwell-time $\mathcal{T}_{2}$ is large enough to satisfy
$\mathcal{T}_{2} \geq \max\left\{ -\frac{1}{\varkappa} \ln\left( \frac{\varsigma_{2}/\varkappa}{\overline{\Psi}_{2} - \varsigma_{2}/\varkappa} \right),\; T_{2,\mathrm{settle}} \right\}$, then it follows that $\mathcal{V}(Z(T_{2}), T_{2}) \leq \underline{\upsilon}(\Upsilon(\iota_{2})) \eqqcolon \overline{\mathcal{V}}_{3} < \overline{\mathcal{V}}_{2}$. Let us define $\theta_{3} \coloneqq \sup_{t \in \mathcal{I}_{3}} \|\theta(t)\|$. Since $\iota_{3} < \iota_{2}$ and class $\mathcal{K}$ functions $\mu^{-1}$ and $\underline{\upsilon}^{-1}$ are strictly increasing, it follows that $\theta_{3} < \theta_{2}$.
Therefore, the gain conditions in \eqref{eq:cond1}-\eqref{eq:cond3}, \eqref{eq:cond4}, and \eqref{eq:cond5} are satisfied over $[T_{2,} T_{3})$, and a Lyapunov-based analysis, similar to \eqref{eq:Lyap2}--\eqref{eq:WBound2}, can be used to show that
\begin{equation}\label{eq:XBound3}
     \|\tilde{x}(t)\| \leq \sqrt{\frac{\lambda_{\max}(P)}{\lambda_{\min}(P)}\left(\max\left\{\overline{\Psi}_{3}, \frac{\varsigma_{3}}{\varkappa}\right\}\right)},
\end{equation}
for all $t \in \mathcal{I}_{3}$, where $\overline{\Psi}_{3} = \frac{2\varsigma_{2}}{\varkappa}$. Since $\varsigma_{3} < \varsigma_{2}$, safety over the interval $\mathcal{I}_{3}$ follows provided $\varepsilon$ from \eqref{eq:epsilonChoice} is used to solve the QP in \eqref{eq:qp}. Given any $\bar{\epsilon}_{x} > 0$, the matrix $S$ in \eqref{eq:cond5} can be chosen with minimum eigenvalue sufficiently large enough such that $\overline{\tilde{x}}_{2} \leq \bar{\epsilon}_{x}$, and as result $\|\tilde{x}(t)\| \leq \bar{\epsilon}_{x}$, for all $t \in \mathcal{I}_{3}$.

The next step is to demonstrate that this analysis can be extended to $[0, \infty)$. If $T_{3} = \infty$, Barbalat’s Lemma \cite[Lemma~8.2]{SCC.Khalil2002} ensures that $\limsup_{t \to \infty}\mathcal{V}(Z(t), t) \leq \underline{\upsilon}(\Upsilon(\iota_{3}))$, where $\iota_{3} = O( \overline{\mathcal{E}}_{3}^2 +\overline{\epsilon}^{2} + \overline{\epsilon}_{\theta}^2)$ and $\overline{\mathcal{E}}_{3} = O(\overline{\tilde{x}}_{2} + \overline{\epsilon}_{\theta})$, which implies that $\limsup_{t \to \infty} Z(t) = O(\overline{\epsilon}_{x} + \overline{\epsilon} + \overline{\epsilon}_{\theta})$.

If $T_{3} < \infty$, the argument extends recursively to subsequent intervals $\mathcal{I}_{s}$ provided the dwell-time satisfies
\begin{equation}\label{eq:finalDwellTime}
        \mathcal{T}_{s} \geq \max\left\{-\frac{1}{\varkappa} \ln\left( 
    \frac{\varsigma_{s} / \varkappa }{ \overline{\Psi}_{s} - \varsigma_{s} / \varkappa }
    \right), T_{s, \operatorname{settle}} \right\},
\end{equation}
so that the gain conditions \eqref{eq:cond1}-\eqref{eq:cond3}, \eqref{eq:cond4}, and \eqref{eq:cond5} hold for all $t > T_{3}$, and the state satisfies 
\begin{equation}\label{eq:XBoundFinal}
    \|\tilde{x}(t)\| \leq \overline{\epsilon}_{x}, \quad \forall t \geq T_{3},
\end{equation}
with $\tilde{x}(t) \in \Omega$ for all $t \in \mathbb{R}_{\geq 0}$. Furthermore, $\overline{\mathcal{E}}_{s} \leq \overline{\mathcal{E}}_{s-1}$ , $\iota_{s} \leq \iota_{s-1}$, $\overline{\mathcal{V}}_{s} \leq \overline{\mathcal{V}}_{s-1}$, and $\theta_{s} \leq \theta_{s-1}$, for all $s > 3$.
 The bound in \eqref{eq:XBoundFinal} and the fact that $\overline{\mathcal{E}}_{s} = O(\overline{\tilde{x}}_{s-1} + \overline{\epsilon}_{\theta})$ indicate $\overline{\mathcal{E}}_{s} = O(\overline{\epsilon}_{x} + \epsilon_{\theta})$ for all $s \in \mathbb{N}$.
This, combined with the dwell time requirement, implies $\limsup_{t\to \infty} \mathcal{V}(Z(t), t) \leq \underline{\upsilon}(\Upsilon(\iota_{s}))$, and hence, $\limsup_{t \to \infty}\|Z(t)\| = O(\overline{\epsilon}_{x} + \overline{\epsilon} + \overline{\epsilon}_{\theta})$. Furthermore, given that $\overline{\tilde{x}}_{s} \leq \overline{\tilde{x}}_{s-1}$ for all $s \in \mathbb{N}$, safety over the interval $[0, \infty)$ is guaranteed when $\varepsilon$ from \eqref{eq:epsilonChoice} is used to solve the QP in \eqref{eq:qp}.
\end{proof}

\section{Simulation Results}\label{section:simulation}
\subsection{Staying Within a Convex Set}\label{sim:convexSet}
To demonstrate how the developed controller can be used to ensure that the trajectories of the system in \eqref{eq:dynamics} stay within a given set despite the lack of full state measurement, the example given in \cite{SCC.Jankovic.ea2018} is considered. The dynamics of the system are of the form in \eqref{eq:dynamics},
with states $x$ = $[x_{1}, x_{2}]^\top$, where 
 \begin{gather}\label{eq:simDyn1}
f(x) =\begin{bmatrix}
-0.6x_{1} - x_{2} \\
 x_{1}^{3} \nonumber
\end{bmatrix}, \
g(x) =  
\begin{bmatrix} 
0 \\ x_{2}
\end{bmatrix} \text{, and }C = \begin{bmatrix}
    1 & 0
\end{bmatrix}. \end{gather}
The control objective is to regulate the trajectory of the system to the origin by minimizing the cost functional in \eqref{eq:costFunctional} while simultaneously ensuring that the trajectory remains confined within the set $\mathcal{S}$ described in \eqref{eq:safeSet1}--\eqref{eq:safeSet3}, as defined by $h(x) = -x_{2}^{2} - x_{1} + 1$. For the cost function in \eqref{eq:cost}, the state and control penalty are selected as $Q = \mathcal{I}_{2}$ and $R = 1$, respectively. The basis for approximation of the value function is selected as $\sigma(x) = 
[x_{1}^{2}, x_{1}x_{2}, x_{2}^{2}]^{\top}$. The actual states, the estimated states, the actor weights, the critic weights, and the least squares gain matrix are initialized as $x(0) = [-2, 1]^{\top}$, $\hat{x}(0) = [-2.5, 1.5]^{\top}$, $\hat{W_{a}}(0) = 0.5\mathbf{1}_{3 \times 1}$, $\hat{W_{c}}(0) = \mathbf{1}_{3 \times 1}$, and $\Gamma (0) = 0.5\mathcal{I}_{3\times3}$, respectively. The learning gains are selected as $k_{a_{1}} = 0.5$, $k_{a_{2}} = 0.1,$ $k_{c} = 5$, $\nu=0.7$, and $\beta = 0.01$. The functions $F(\hat{x})$, $G_{i}^{-}(\hat{x})$, and $G_{i}^{+}(\hat{x})$ are determined using Lipschitz constants on $\bar{B}(\hat{x}, \varepsilon)$ with $\varepsilon = 0.7$, all of which are assumed to be $0.2$. The simulation uses 100 fixed Bellman error extrapolation points selected from a $2\times 2$ square centered around the origin of the system. The QP optimization problem in Problem~\ref{prob:coop} is solved using \texttt{quadprog} on MATLAB to obtain the safe controller in \eqref{eq:qp}.
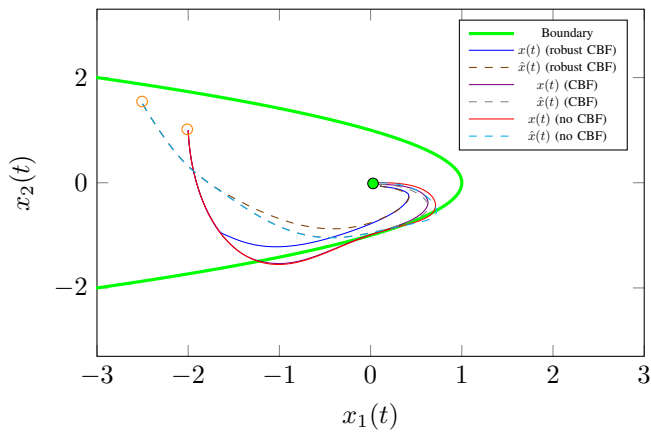
\begin{figure}
        \centering
        \begin{tikzpicture} 
\begin{axis}[
    xlabel={$x_{1}(t)$},
    ylabel={$x_{2}(t)$},
    xmin=-3, xmax=3,
    ymin=-3, ymax=3,
    legend pos = north east,
    legend style={nodes={scale=0.5, transform shape}},
    enlarge y limits=0.05,
    width=\linewidth,
    height=0.7\linewidth,
    view={0}{90},
]

\addplot[
    domain=-3:3,
    samples=400,
    very thick,
    green
] ({1 - x^2}, {x});
\addlegendentry{Boundary};

\addplot [color = blue] table {data/xDataRCBFSim1.dat};
\addlegendentry{$x(t)$ (robust CBF)};

\addplot [dashed, color = orange!50!black] table {data/xHatDataRCBFSim1.dat};
\addlegendentry{$\hat{x}(t)$ (robust CBF)};

\addplot [color = violet] table {data/xDataCBFSim1.dat};
\addlegendentry{$x(t)$ (CBF)};

\addplot [dashed, color = gray] table {data/xHatDataCBFSim1.dat};
\addlegendentry{$\hat{x}(t)$ (CBF)};

\addplot [color = red] table {data/xDataNoCBFSim1.dat};
\addlegendentry{$x(t)$ (no CBF)};

\addplot [color = cyan, dashed] table {data/xHatDataNoCBFSim1.dat};
\addlegendentry{$\hat{x}(t)$ (no CBF)};

\end{axis}

\filldraw[fill=green] (3.67,2.3) circle[radius=2pt];
\draw[orange] (1.2,3.02) circle[radius=2pt];
\draw[orange] (0.6,3.39) circle[radius=2pt];

\end{tikzpicture}
	\caption{Results for Section~\ref{sim:convexSet}, safety within a given set (marked by a thick green boundary). The trajectories of the actual states $x$ and the estimated state $\hat{x}$ under the safe controller in \eqref{eq:qp} (robust CBF) are compared against controllers with standard CBF and without CBF.}
		\label{fig:safetySim1}
\end{figure}
\begin{figure}
        \centering
        \begin{tikzpicture}
    \begin{axis}[
        xlabel={$t$ [s]},
        ylabel={$\tilde{x}(t)$},
        legend pos = north east,
        legend style={nodes={scale=0.8, transform shape}},
        enlarge y limits=0.05,
        enlarge x limits=0,
        width=\linewidth,
        height=0.45\linewidth,
    ]
    \pgfplotsinvokeforeach{1,...,2}{
        \addplot+ [mark=none] table [x index=0, y index=#1] {data/eDataSim1.dat};
    }
    \legend{$\tilde{x}_{1}(t)$, $\tilde{x}_{2}(t)$}
    \end{axis}
\end{tikzpicture}
	
		\caption{Estimation errors between the actual states and the estimated states for the experiment in Section~\ref{sim:convexSet}.}
		\label{fig:errorSim1}
\end{figure}
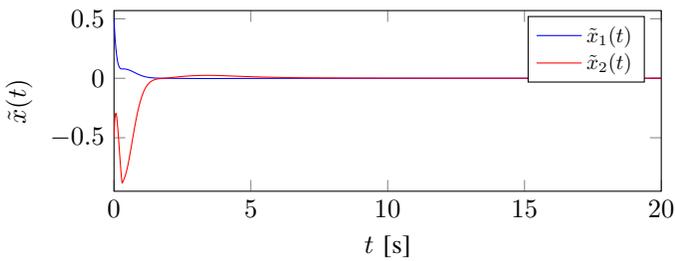

For the DNN-based state observer in \eqref{eq:DNNDynamicsEstimate}, the matrix $A = [-0.6, -1; 0, 0]$ and the linear correction gains $K = [10.4, -30]^{\top}$ are selected via pole placement with poles at $[-5, -6]$ so that $A-KC$ is Hurwitz. A DNN architecture with three inner layers, each comprising of $10$, $6$, and $7$ neurons, respectively, is employed for state estimation using data collected from a simulated two-state system with varying initial conditions and using a random control policy, generating input-output pairs $\left\{\hat{x}(t_{j}), \dot{\hat{x}}(t_{j})\right\}$ for all $t_{j} \in \mathbb{R}_{\geq 0}$ and $j \in \mathbb{Z}_{\geq 0}$. The activation functions of the DNN are Elliot symmetric sigmoid, logarithmic sigmoid, and tangent sigmoid for the respective layers. The Levenberg-Marquardt algorithm in the \texttt{Deep Learning Toolbox} on MATLAB trains the network on 70\% training, 15\% validation, and 15\% testing data to minimize the MSE to $5 \times 10^{-3}$ within 10,000 epochs. The simulation parameters and gains are selected as $k_{\theta} = 100$, $\gamma = \mathcal{I}_{13 \times 13}$, $\kappa = 0.5$, $\Delta t = 0.25$, and the history stack size is capped at $M = 20$. At the start of the simulation, the history stacks and the initial outer layer weights are initialized to zero. These weights are subsequently updated according to the law in \eqref{eq:outerLayerWeightsUpdate}, and the inner DNN weights are trained every $2$ seconds according to Algorithm~\ref{algo:parameterEstimatorAlgo}.
\begin{figure}
        \centering
       \begin{tikzpicture}
    \begin{axis}[
        xlabel={$t$ [s]},
        ylabel={$\hat{W}_{c}(t)$},
        legend pos = north east,
        legend style={nodes={scale=0.75, transform shape}},
        enlarge y limits=0.01,
        enlarge x limits=0,
        width=\linewidth,
        height=0.45\linewidth,
    ]
    \pgfplotsinvokeforeach{1,...,3}{
        \addplot+ [solid, color=color#1, mark=none] table [x index=0, y index=#1] {data/WcHatDataSim1.dat};
    }
    \legend{$\hat{W}_{c_{1}}(t)$, $\hat{W}_{c_{2}}(t)$, $\hat{W}_{c_{3}}(t)$}
    \end{axis}
\end{tikzpicture}
		\caption{Estimated critic weights for the experiment in Section~\ref{sim:convexSet}.}
		\label{fig:valueWeightsSim1}
\end{figure}
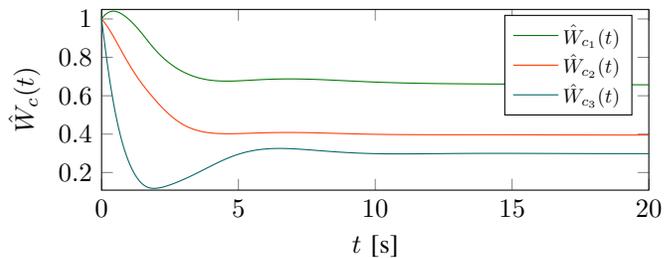
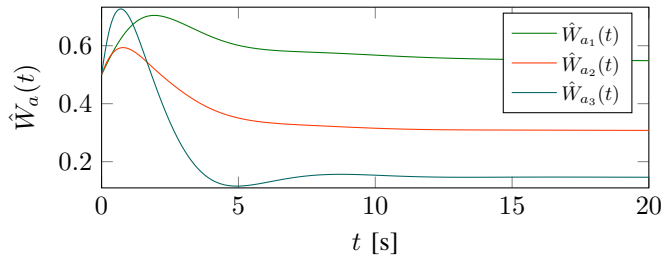
\begin{figure}
        \centering
       \begin{tikzpicture}
    \begin{axis}[
        xlabel={$t$ [s]},
        ylabel={$\hat{W}_{a}(t)$},
        legend pos = north east,
        legend style={nodes={scale=0.75, transform shape}},
        enlarge y limits=0.01,
        enlarge x limits=0,
        width=\linewidth,
        height=0.45\linewidth,
    ]
    \pgfplotsinvokeforeach{1,...,3}{
        \addplot+ [solid, color=color#1, mark=none] table [x index=0, y index=#1] {data/WaHatDataSim1.dat};
    }
    \legend{$\hat{W}_{a_{1}}(t)$, $\hat{W}_{a_{2}}(t)$, $\hat{W}_{a_{3}}(t)$}
    \end{axis}
\end{tikzpicture}
		\caption{Estimated actor weights for the experiment in Section~\ref{sim:convexSet}.}
		\label{fig:policyWeightsSim1}
\end{figure}
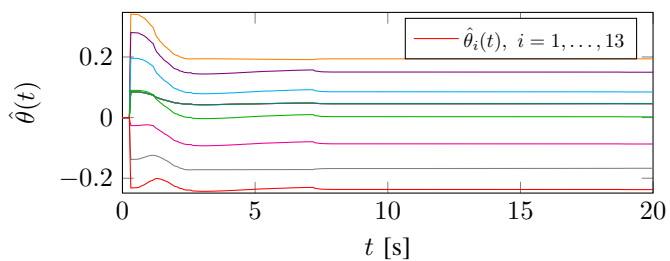
\begin{figure}
        \centering
       \begin{tikzpicture}
    \begin{axis}[
        xlabel={$t$ [s]},
        ylabel={$\hat{\theta}(t)$},
        legend pos = north east,
        legend style={nodes={scale=0.75, transform shape}},
        enlarge y limits=0.01,
        enlarge x limits=0,
        width=0.975\linewidth,
        height=0.45\linewidth,
        cycle list name=color list,
    ]
    \pgfplotsinvokeforeach{1,...,13}{
        \addplot+ [solid, mark=none] table [x index=0, y index=#1] {data/thetaHatDataSim1.dat};
    }
    \addlegendentry{$\hat{\theta}_i(t),\ i=1,\dots,13$}
    \end{axis}
\end{tikzpicture}
		\caption{Estimated outer layer DNN weights for the experiment in Section~\ref{sim:convexSet}.}
		\label{fig:parameterWeightsSim1}
\end{figure}

\subsection{Obstacle Avoidance}\label{sim:collision}
To demonstrate how the developed controller can be used to avoid an obstacle, a simulation study is performed on a control-affine nonlinear system obtained from \cite{SCC.Ren.Zhang.ea2023} which is of the form in \eqref{eq:dynamics} with states $x$ = $[x_{1}, x_{2}]^\top$, where
 \begin{gather}\label{eq:simDyn2}
f(x) =\begin{bmatrix}
-x_{1} - x_{2} \\
 -\frac{1}{2}x_{1}-\frac{1}{2}x_{2}\big(1-x_{1}^{2}\big)-x_{1}^{2}x_{2} \nonumber
\end{bmatrix}, \\
g(x) =  \begin{bmatrix}0 \\
\cos\left(2x_{1}\right)+2
\end{bmatrix},\text{ and }C = \begin{bmatrix}
    1 & 0
\end{bmatrix}. \end{gather}
 The control objective is to minimize the infinite horizon cost in \eqref{eq:costFunctional} while avoiding a circular obstacle. The safe set $\mathcal{S}$ described in \eqref{eq:safeSet1}--\eqref{eq:safeSet3} is defined by   
    $h(x) = \|x - z\| - r$, 
 where $z = [-0.7, 1.2]^{\top}$ is the center and $r = 0.35$ is the radius. The actual states, the estimated states, the actor weights, the critic weights, and the least squares gain matrix are initialized as $x(0) = [-0.5, 2]^{\top}$, $\hat{x}(0) = [-0.75, 2.25]^{\top}$, $\hat{W_{a}}(0) = 0.5\mathbf{1}_{3 \times 1}$, $\hat{W_{c}}(0) = 0.5\mathbf{1}_{3 \times 1}$, and $\Gamma (0) = \mathcal{I}_{3\times3}$, respectively. For the DNN-based state observer in \eqref{eq:DNNDynamicsEstimate}, the matrix $A = [-1, -1; -0.5, -0.5]$ and the linear correction gain $K = [-13, 5]^{\top}$ is selected via pole placement with poles at $[-3, -4]$ so that $A-KC$ is Hurwitz. The learning gains are selected as $k_{a_{1}} = 1$, $k_{a_{2}} = 0.5$, and $k_{c} = 0.5$. The functions $F(\hat{x})$, $G_{i}^{-}(\hat{x})$, and $G_{i}^{+}(\hat{x})$ are determined using Lipschitz constants on $\bar{B}(\hat{x}, \varepsilon)$ with $\varepsilon = 0.5$, all of which are assumed to be $0.1$. The initial outer layer DNN weights are also selected as $\theta(0) = 0_{13 \times 13}$ and updated according to the law in \eqref{eq:outerLayerWeightsUpdate}, and the inner DNN weights are trained every $2$ seconds according to Algorithm~\ref{algo:parameterEstimatorAlgo}. The remaining simulation parameters and gains are selected similarly to the first study in Section~\ref{sim:convexSet}.
\begin{figure}
        \centering
        \begin{tikzpicture}
    \begin{axis}[
        axis equal,
        xlabel={$x_{1}(t)$},
        ylabel={$x_{2}(t)$},
        legend pos=north east,
        legend style={nodes={scale=0.6, transform shape}},
        width=\linewidth,
        height=0.7\linewidth,
        legend cell align={center},
    ]

        \def\zOne{-0.7} 
        \def\zTwo{1.2}   
        \def\radius{0.35} 
        
        \draw[thin, red, fill=red!40, opacity=0.5] 
            (\zOne, \zTwo) circle (\radius);
        \fill[red] (\zOne, \zTwo) circle (1.5pt);
        
        \draw[red!80] (\zOne+0.1, \zTwo) -- (\zOne+0.5, \zTwo+0.05);
        \node[below right, red!80, font=\footnotesize] at (\zOne+0.5, \zTwo+0.175) {Obstacle};
        
        \addplot [color = blue] table {data/xDataRCBFSim2.dat};
        \addlegendentry{$x(t)$ (robust CBF)};
        
        \addplot [dashed, color = orange!50!black] table {data/xHatDataRCBFSim2.dat};
        \addlegendentry{$\hat{x}(t)$ (robust CBF)};
        
        \addplot [color = violet] table {data/xDataCBFSim2.dat};
        \addlegendentry{$x(t)$ (CBF)};
        
        \addplot [dashed, color = gray] table {data/xHatDataCBFSim2.dat};
        \addlegendentry{$\hat{x}(t)$ (CBF)};

        \addplot [color = red] table {data/xDataNoCBFSim2.dat};
        \addlegendentry{$x(t)$ (no CBF)};
        
        \addplot [dashed, color = cyan] table {data/xHatDataNoCBFSim2.dat};
        \addlegendentry{$\hat{x}(t)$ (no CBF)};

    \end{axis}

    \filldraw[fill=green] (4.58, 0.4) circle[radius=1.5pt];
    
    \draw[orange] (3.73, 3.8) circle[radius=1.5pt];
    \draw[orange] (3.305, 4.23) circle[radius=1.5pt];
\end{tikzpicture}
	\caption{Results for Section~\ref{sim:collision}, obstacle avoidance. The trajectories of the actual states $x$ and the estimated state $\hat{x}$ under the safe controller in \eqref{eq:qp} (robust CBF) are compared against controllers with standard CBF and without CBF.}
		\label{fig:barrierSim2}
\end{figure}
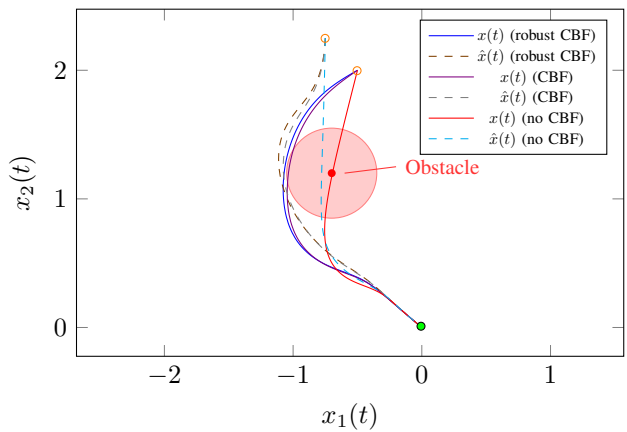
\begin{figure}
        \centering
        \begin{tikzpicture}
    \begin{axis}[
        xlabel={$t$ [s]},
        ylabel={$\tilde{x}(t)$},
        legend pos = north east,
        legend style={nodes={scale=0.75, transform shape}},
        enlarge y limits=0.01,
        enlarge x limits=0,
        width=\linewidth,
        height=0.45\linewidth,
    ]
    \pgfplotsinvokeforeach{1,...,2}{
        \addplot+ [mark=none] table [x index=0, y index=#1] {data/eDataSim2.dat};
    }
    \legend{$\tilde{x}_{1}(t)$, $\tilde{x}_{2}(t)$}
    \end{axis}
\end{tikzpicture}
	
		\caption{Estimation errors between the actual states and the estimated states for the experiment in Section~\ref{sim:collision}.}
		\label{fig:errorSim2}
\end{figure}
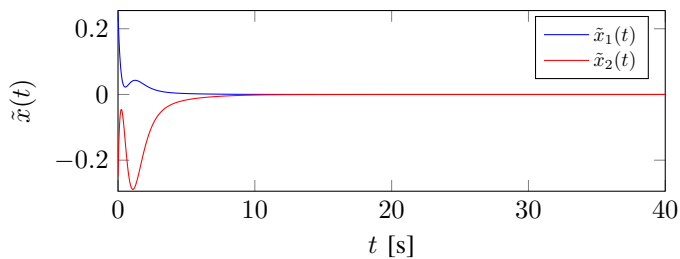

\subsection{Discussion of results}
It can be observed from Fig.~\ref{fig:safetySim1} that the developed controller in \eqref{eq:qp} guarantees safety despite partial state measurement in the simulation study in \ref{sim:convexSet}. As shown in Fig.~\ref{fig:safetySim1}, without the robust CBF, the trajectory of the system crosses the green boundary of the safe set. 
Similarly, for the obstacle avoidance problem in Section~\ref{sim:collision}, without the robust CBF, the system fails to avoid the obstacle. Conversely, the safe controller in \eqref{eq:qp} guarantees that system trajectories steer clear of the obstacle while converging to the origin.
 From Figs.~\ref{fig:errorSim1} and \ref{fig:errorSim2}, it can be seen that the state estimation errors converge to a neighborhood of the origin, which showcases the effectiveness of the designed DNN-based observer. The actor and critic weight estimates, as depicted in Figs.~\ref{fig:valueWeightsSim1}, \ref{fig:policyWeightsSim1}, \ref{fig:policyWeightsSim2}, and \ref{fig:valueWeightsSim2}, are shown to be bounded, which correlates with the results of Theorem \ref{thm:boundedAndSafe}. 
 
 For DNN-based state estimation, data are initially gathered, and weights and biases are trained. After approximately $2$ seconds, sufficient data are available for retraining the DNN. In this phase, weights are updated, and the history stack $\mathcal{H}$ is purged and repopulated according to the minimum eigenvalue maximization algorithm in Algorithm~\ref{algo:parameterEstimatorAlgo}. The first half of the simulation is dedicated to data collection and training. Following retraining, improvements in state estimates are observed and the convergent outer layer weights are shown in Figs.~\ref{fig:parameterWeightsSim1} and \ref{fig:parameterWeightsSim2}.
 \begin{figure}
        \centering
    \begin{tikzpicture}
    \begin{axis}[
        xlabel={$t$ [s]},
        ylabel={$\hat{W}_{c}(t)$},
        legend pos = north east,
        legend style={nodes={scale=0.75, transform shape}},
        enlarge y limits=0.05,
        enlarge x limits=0,
        width=\linewidth,
        height=0.45\linewidth,
    ]
    \pgfplotsinvokeforeach{1,...,3}{
        \addplot+ [solid, color=color#1, mark=none] table [x index=0, y index=#1] {data/WcHatDataSim2.dat};
    }
    \legend{$\hat{W}_{c_{1}}(t)$, $\hat{W}_{c_{2}}(t)$, $\hat{W}_{c_{3}}(t)$}
    \end{axis}
\end{tikzpicture}
		\caption{Estimated critic weights for the experiment in Section~\ref{sim:collision}.}
		\label{fig:valueWeightsSim2}
\end{figure}
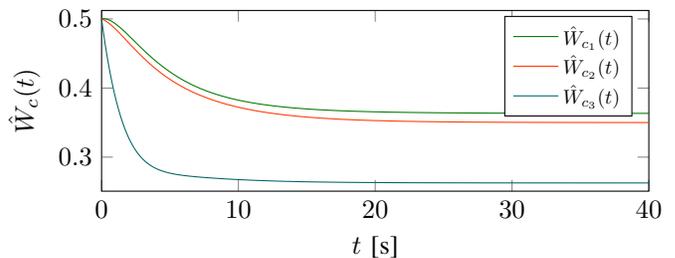
\begin{figure}
        \centering
    \begin{tikzpicture}
    \begin{axis}[
        xlabel={$t$ [s]},
        ylabel={$\hat{W}_{a}(t)$},
        legend pos = north east,
        legend style={nodes={scale=0.75, transform shape}},
        enlarge y limits=0.05,
        enlarge x limits=0,
        width=\linewidth,
        height=0.45\linewidth,
    ]
    \pgfplotsinvokeforeach{1,...,3}{
        \addplot+ [solid, color=color#1, mark=none] table [x index=0, y index=#1] {data/WaHatDataSim2.dat};
    }
    \legend{$\hat{W}_{a_{1}}(t)$, $\hat{W}_{a_{2}}(t)$, $\hat{W}_{a_{3}}(t)$}
    \end{axis}
\end{tikzpicture}
		\caption{Estimated actor weights for the experiment in Section~\ref{sim:collision}.}
		\label{fig:policyWeightsSim2}
\end{figure}
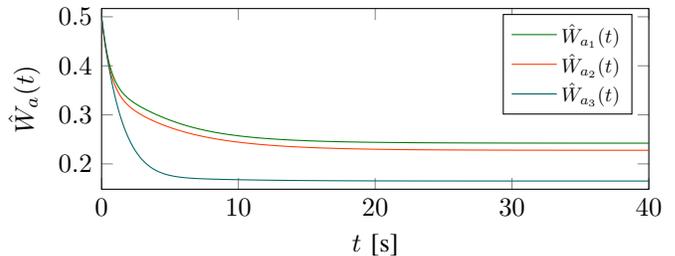
\begin{figure}
        \centering
       \begin{tikzpicture}
    \begin{axis}[
        xlabel={$t$ [s]},
        ylabel={$\hat{\theta}(t)$},
        legend pos = north east,
        legend style={nodes={scale=0.75, transform shape}},
        enlarge y limits=0.01,
        enlarge x limits=0,
        width=0.975\linewidth,
        height=0.45\linewidth,
        cycle list name=color list,
        scaled y ticks=false,  
        yticklabel style={/pgf/number format/fixed, /pgf/number format/precision=2},
    ]
    \pgfplotsinvokeforeach{1,...,13}{
        \addplot+ [solid, mark=none] table [x index=0, y index=#1] {data/thetaHatDataSim2.dat};
    }
    \addlegendentry{$\hat{\theta}_i(t),\ i=1,\dots,13$}
    \end{axis}
\end{tikzpicture}
		\caption{Estimated outer layer DNN weights for the experiment in Section~\ref{sim:collision}.}
		\label{fig:parameterWeightsSim2}
\end{figure}
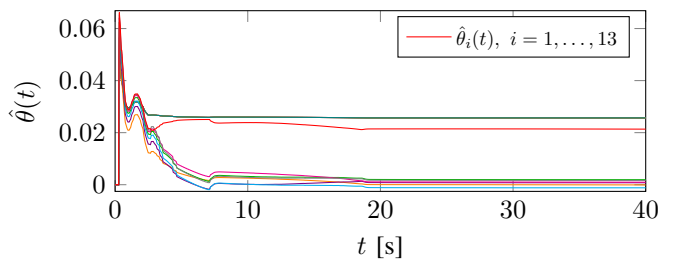
\section{Conclusion}\label{section:conclusion}
This paper tackles safety in output feedback systems by proposing an SAOC framework that utilizes state and parameter estimation, along with CBFs, to enable safe adaptive control. The objective is to realize online, real-time safe learning that guarantees convergence and safety despite uncertainty in the state. The developed approach separates the infinite-horizon optimal control problem from the CBF-based safety problem. The safe controller is then integrated with the SAOC-based stabilizing control policy. This integration results in a controller that is both approximately optimal and robust against state estimation errors, thereby ensuring effective real-time adaptation and learning while achieving safety objectives.

 The integration of the DNN framework within a minimum eigenvalue maximization algorithm, featuring real-time adaptation of outermost layer weights and batch updates for inner layers, provides an effective method of state estimation as evidenced by the simulation results in Section~\ref{section:simulation}. While the simulation results demonstrate the effectiveness of the technique developed in this paper, adherence to safety constraints relies on the choice of the $\varepsilon$-bound on the state estimation error, which as derived from the analysis in Section~\ref{section:stabilityAnalysis}, requires knowledge of bounds on the true weights of the optimal value function in \eqref{eq:optimalV}. Since bounds on the optimal value function are unknown, we must rely on heuristic estimates. In the simulation examples, $\varepsilon$ is determined based on such estimates, which, lacking precise guarantees, can lead to overly conservative or even incorrect choices of $\varepsilon$. Assuming that the optimal value function lies in a reproducing kernel Hilbert space (RKHS), we can derive explicit bounds on the critic's performance in terms of kernel functions, the number of basis functions, and the scattered location of centers used to define the RKHS, using the approach in \cite{SCC.Niu.Bouland.ea2024}. The bounds on the performance of the critic will allow for a better selection of $\varepsilon$. Furthermore, the safe control policy in \eqref{eq:qp} relies on Lipschitz bounds of the system dynamics in \eqref{eq:dynamics}. Since these bounds are local, estimating them accurately over the entire operating region is challenging and may lead to conservative behavior or safety violations. Thus, the effectiveness of the approach in this paper depends on how well these bounds capture the true dynamics.

Future research will focus on deriving explicit bounds on the value function to obtain more accurate state estimation error bounds. Furthermore, the developed approach will be extended to multiagent output feedback RL where each agent solves an output-feedback adaptive optimal control problem to realize a control framework that is motivated by the ``centralized learning, decentralized execution'' paradigm of multiagent RL \cite{SCC.Lowe.Wu.ea2017}.

\small
\bibliographystyle{IEEETrans.bst}
\bibliography{scc,sccmaster,scctemp}
\end{document}